\documentclass[11pt]{article}
\usepackage[margin=1in]{geometry}
\usepackage{times}
\usepackage[compact]{titlesec}
\usepackage{amsfonts}
\usepackage{amsmath,amsthm,amssymb}
\usepackage{latexsym}
\usepackage{epic}
\usepackage{epsfig}
\usepackage{hyperref}
\usepackage{verbatim}
\usepackage[justification=centering]{caption}
\usepackage{enumitem}
\usepackage{color}
\usepackage[numbers]{natbib}
\usepackage{algorithm}
\usepackage{algorithmic}
\usepackage{setspace}
\usepackage[labelformat=simple]{subcaption}

\usepackage{float}
\newfloat{algorithm}{t}{lop}

\setlist[enumerate]{topsep=4pt, itemsep=3pt}

\newtheorem{theorem}{Theorem}[section]
\newtheorem{corollary}[theorem]{Corollary}
\newtheorem{lemma}[theorem]{Lemma}

\newtheorem{claim}[theorem]{Claim}
\newtheorem{definition}[theorem]{Definition}
\newtheorem{remark}[theorem]{Remark}

\newtheorem*{conjecture*}{Conjecture}
\newtheoremstyle{nonindented}{1ex}{1ex}{}{}{\bfseries}{.}{.5em}{}
\newtheoremstyle{indented}{1ex}{1ex}{\itshape\addtolength{\leftskip}{0.6cm}\addtolength{\rightskip}{0.6cm}}{}{\bfseries}{.}{.5em}{}
\theoremstyle{nonindented}
\theoremstyle{indented}
\theoremstyle{plain}

\newenvironment{alg}{\begin{algorithm}\begin{onehalfspace}\begin{algorithmic}[1]}{\end{algorithmic}\end{onehalfspace}\end{algorithm}}

\newcommand{\cross}{\times}

\newcommand{\bvec}[1]{\boldsymbol{ #1 }}

\newcommand{\set}[1]{\left\{ #1 \right\}}
\newcommand{\union}{\cup}

\newcommand{\ceil}[1]{\lceil {#1} \rceil}
\renewcommand{\hat}{\widehat}
\renewcommand{\tilde}{\widetilde}
\renewcommand{\bar}{\overline}

\DeclareMathOperator{\poly}{poly}

\def\max{\qopname\relax n{max}}

\def\argmax{\qopname\relax n{argmax}}

\def\Pr{\qopname\relax n{\mathbf{Pr}}}
\def\Ex{\qopname\relax n{\mathbf{E}}}

\newcommand{\RR}{\mathbb{R}}

\def\A{\mathcal{A}}
\def\B{\mathcal{B}}
\def\C{\mathcal{C}}
\def\D{\mathcal{D}}

\def\K{\mathcal{K}}

\def\M{\mathcal{M}}
\def\P{\mathcal{P}}

\def\O{\mathcal{O}}

\def\eps{\epsilon}

\def\SS{\mathbb{S}}

\newcommand{\eat}[1]{}

\newcommand{\INPUT}{\item[\textbf{Input:}]}
\newcommand{\OUTPUT}{\item[\textbf{Output:}]}
\newcommand{\PARAMETER}{\item[\textbf{Parameter:}]}

\newcommand{\maxi}[1]{\mbox{maximize} & {#1 } & \\}
\newcommand{\maxis}[1]{\mbox{max} & {#1 } & \\}

\newcommand{\st}{\mbox{subject to} }
\newcommand{\sts}{\mbox{s.t.} }
\newcommand{\con}[1]{&#1 & \\}
\newcommand{\qcon}[2]{&#1, & \mbox{for } #2.  \\}
\newenvironment{lp}{\begin{equation}  \begin{array}{lll}}{\end{array}\end{equation} }
\newenvironment{lp*}{\begin{equation*}  \begin{array}{lll}}{\end{array}\end{equation*}}

\begin{document}

\title{Algorithmic Bayesian Persuasion}

\author{
Shaddin Dughmi\thanks{Supported in part by NSF CAREER Award CCF-1350900.} \\
Department of Computer Science\\
University of Southern California\\
{\tt shaddin@usc.edu}
\and
Haifeng Xu\thanks{Supported by NSF grant CCF-1350900.} \\
Department of Computer Science\\
University of Southern California\\
{\tt haifengx@usc.edu}
}


\begin{titlepage}
\clearpage\maketitle
\thispagestyle{empty}

\begin{abstract}
	
	\emph{Persuasion}, defined as the act of exploiting an informational advantage in order to effect the decisions of others, is ubiquitous. Indeed, persuasive communication has been estimated to account for almost a third of all economic activity in the US. This paper examines persuasion through a computational lens, focusing on what is perhaps the most basic and fundamental model in this space: the celebrated \emph{Bayesian persuasion} model of Kamenica and Gentzkow~\cite{Kamenica11}. Here there are two players, a \emph{sender} and a \emph{receiver}. The receiver must take one of a number of actions with a-priori unknown payoff, and the sender has access to additional information regarding the payoffs of the various actions for both players. The sender can commit to revealing a noisy signal regarding the realization of the payoffs of various actions, and would like to do so as to maximize her own payoff  in expectation assuming that the receiver rationally acts to maximize his own payoff. When the payoffs of various actions follow a joint distribution (the common prior), the sender's problem is nontrivial, and its computational complexity depends on the representation of this prior.

	We examine  the sender's optimization task  in three of the most natural input models for this problem, and essentially pin down its computational complexity in each. When the payoff distributions of the different actions are i.i.d.\ and given explicitly, we exhibit a polynomial-time (exact) algorithmic solution, and a ``simple'' $(1-1/e)$-approximation algorithm. Our optimal scheme for the i.i.d.\ setting involves an analogy to auction theory, and makes use of Border's characterization of  the space of reduced-forms for single-item auctions. 
	When action payoffs are independent but non-identical with marginal distributions given explicitly, we show that  it is \#P-hard to compute the optimal expected sender utility. In doing so, we rule out a \emph{generalized Border's theorem}, as defined by Gopalan et al~\cite{Gopalan15}, for this setting.  Finally, we consider a general (possibly correlated) joint distribution of action payoffs presented by a black box sampling oracle, and exhibit a fully polynomial-time approximation scheme (FPTAS)  with a bi-criteria guarantee. Our FPTAS is based on Monte-Carlo sampling, and its analysis relies on the principle of deferred decisions. Moreover, we show that this result is the best possible in the black-box model for information-theoretic reasons. 

\end{abstract}

\end{titlepage}

\newpage

\section{Introduction}


\begin{quote}
	``\it  One quarter of the GDP is persuasion.''
\end{quote}

This is both the title, and the thesis, of a 1995 paper by McCloskey and Klamer \cite{McCloskey95}. Since then, persuasion as a share of economic activity appears to be growing --- a more recent estimate places the figure at 30\% \cite{antioch}. As both papers make clear, persuasion is intrinsic in most human endeavors.  When the tools of ``persuasion'' are tangible  --- say goods, services, or money --- this is the domain of traditional \emph{mechanism design}, which steers the actions of one or many self-interested agents towards a designer's objective. What \cite{McCloskey95, antioch} and much of the relevant literature refer to as persuasion, however, are scenarios in which the power to persuade derives from an \emph{informational advantage} of some party over others. This is also the sense in which we use the term. Such scenarios are increasingly common in the information economy, and it is therefore unsurprising that  persuasion has been the subject of a large body of work in recent years, motivated by domains as varied as auctions~\cite{Bergemann07,Eso07,Emek12,Bergemann14}, advertising  \cite{Anderson06,Johnson06,Chakraborty12}, voting \cite{Alonso14},  security \cite{Xu15,Rabinovich15}, multi-armed bandits \cite{crowds,bandit_exploration}, medical research \cite{kolotilin_experimental}, and financial regulation \cite{stresstest1,stresstest2}. (For an empirical survey of persuasion, we refer the reader to \cite{Dellavigna10}). What is surprising, however, is the lack of systematic study of persuasion through a computational lens; this is what we embark on in this paper.

In the large body of literature devoted to persuasion, perhaps no model is more basic and fundamental than the \emph{Bayesian Persuasion} model of Kamenica and Gentzkow \cite{Kamenica11}, generalizing an earlier model by Brocas and Carrillo \cite{Brocas07}. Here there are two players, who we call the \emph{sender} and the \emph{receiver}.  The receiver is faced with selecting one of a number of \emph{actions}, each of which is associated with an a-priori unknown payoff to both players. The \emph{state of nature}, describing the payoff to the sender and receiver from each action, is drawn from a prior distribution known to both players. However, the sender possesses an informational advantage, namely access to the \emph{realized} state of nature prior to the receiver choosing his action. In order to persuade the receiver to take a more favorable action for her, the sender can \emph{commit} to a policy, often known as an \emph{information structure} or \emph{signaling scheme}, of releasing information about the realized state of nature to the receiver before the receiver makes his choice.  This policy may be simple, say by always announcing the payoffs of the various actions or always saying nothing, or it may be intricate, involving partial information and added noise. Crucially, the receiver is aware of the sender's committed policy, and moreover is rational and Bayesian. We examine the sender's algorithmic problem of implementing the optimal signaling scheme in this paper. A solution to this problem, i.e., a signaling scheme, is an algorithm which takes as input the description of a state of nature and outputs a signal, potentially utilizing some internal randomness.


\subsection{Two Examples}

To illustrate the intricacy of Bayesian Persuasion, \citet{Kamenica11} use a simple example in which the sender is a prosecutor, the receiver is a judge, and the state of nature is the guilt or innocence of a defendant. The receiver (judge) has two actions, conviction and acquittal, and wishes to maximize the probability of rendering the correct verdict. On the other hand, the sender (prosecutor) is interested in maximizing the probability of conviction. As they show, it is easy to construct examples in which the optimal signaling scheme for the sender releases noisy partial information regarding the guilt or innocence of the defendant. For example, if the defendant is guilty with probability $\frac{1}{3}$, the prosecutor's best strategy is to claim ``guilt" whenever the defendant is guilty, and also claim ``guilt" just under half the time when the defendant is innocent. 
As a result, the defendant will be convicted  whenever the prosecutor claims ``guilt'' (happening with probability just under $\frac{2}{3}$), assuming that the judge is fully aware of the prosecutor's signaling scheme. We note that it is not in the prosecutor's interest to always claim ``guilt'', since a rational judge aware of such a policy would ascribe no meaning to such a signal, and render his verdict based solely on his prior belief --- in this case, this would always lead to acquittal.\footnote{In other words, a signal is an abstract object with no intrinsic meaning, and is only imbued with meaning by virtue of how it is used. In particular, a signal has no meaning beyond the posterior distribution on states of nature it induces.}

A somewhat less artificial example of persuasion is in the context of providing financial advice. Here, the receiver is an investor,  actions correspond to stocks, and the sender is a stockbroker or financial adviser with access to stock return projections which are a-priori unknown to the investor. When the adviser's commission or return is not aligned with the investor's returns, this is a nontrivial Bayesian persuasion problem. In fact, interesting examples exist when stock returns are independent from each other, or even i.i.d. Consider the following simple example which fits into the i.i.d. model considered in Section \ref{sec:iid}: there are two stocks, each of which is a-priori equally likely to generate  low (L), moderate (M), or high (H) short-term returns to the investor (independently). We refer to L/M/H as the \emph{types} of a stock, and associate them with short-term returns of $0$, $1+\eps$, and $2$ respectively. Suppose, also, that stocks of type L or H are associated with poor long-term returns of $0$; in the case of H, high short-term returns might be an indication of volatility or overvaluation, and hence poor long-term performance. This leaves stocks of type M as the only solid performers with long-term returns of $1$. Now suppose that the investor is myopically interested in maximizing short-term returns, whereas the forward-looking financial adviser is concerned with maximizing long-term returns, perhaps due to reputational considerations. Simple calculation shows that providing full information to the myopic investor results in an expected long-term reward of $\frac{1}{3}$, as does providing no information. An optimal signaling scheme, which guarantees that the investor chooses a stock with type M whenever such a stock exists, is the following: when exactly one of the stocks has type M recommend that stock, and otherwise recommend a stock uniformly at random. A simple calculation using Bayes' rule shows that the investor prefers to follow the recommendations of this partially-informative scheme, and it follows that the expected long-term return is $\frac{5}{9}$. 

\subsection{Results and Techniques}
\label{sec:intro:results}
Motivated by these intricacies, we study the computational complexity of  optimal and near-optimal persuasion in the presence of multiple actions. We first observe that a linear program with a variable for each (state-of-nature, action) pair computes a description of the optimal signaling scheme. However, when action payoffs are distributed according to a joint distribution --- say exhibiting some degree of independence across different actions --- the number of states of nature may be exponential in the number of actions; in such settings, both the number of  variables and constraints of this linear program are exponential in the number of actions. It is therefore unsurprising that the computational complexity of persuasion depends on how the prior distribution on states of nature is presented as input. We therefore consider three natural input models in increasing order of generality, and mostly pin down the complexity of optimal and near-optimal persuasion in each. Our first model assumes that action payoffs are drawn i.i.d. from an explicitly described marginal distribution. Our second model considers independent yet non-identical actions, again with explicitly-described marginals. Our third and most general model considers an arbitrary joint distribution of action payoffs presented by a black-box sampling oracle. 
In proving our results, we draw connections to techniques and concepts developed in the context of Bayesian mechanism design (BMD), exercising and generalizing them along the way as needed to prove our results. We mention some of these connections briefly here, and elaborate on the similarities and differences from the BMD literature in Appendix \ref{app:intro:related_bmd}. 

We start with the i.i.d model, and show two results: a ``simple'' and polynomial-time $\frac{e-1}{e}$-approximate signaling scheme, and a polynomial-time  implementation of the optimal scheme. Both results hinge on a ``symmetry characterization'' of the optimal scheme in the i.i.d. setting, closely related  to the symmetrization result from BMD by~\cite{weinberg_symmetry} but with an important difference which we discuss in Appendix~\ref{app:intro:related_bmd}. Our ``simple'' scheme  decouples the signaling problem for the different actions and signals \emph{independently} for each. This result implies that  signaling in this setting can be ``distributed'' among multiple non-coordinating persuaders without much loss. 
Our optimal scheme  involves a connection to Border's characterization of the space of feasible reduced-form auctions \cite{Border91,Border07}, as well as its algorithmic properties~\cite{Cai12a, Saeed12}. This connection involves proving  a correspondence between ``symmetric'' signaling schemes and a subset of ``symmetric'' single-item auctions; one in which actions in persuasion correspond to bidders in an auction.  

Next, we consider Bayesian persuasion with independent non-identical actions. One might expect that the partial correspondence between signaling schemes and single-item auctions in the i.i.d. model generalizes here, in which case Border's theorem --- which extends to single-item auctions with independent non-identical bidders ---  would analogously  lead to polynomial time algorithm for persuasion in this setting. However, we surprisingly show that this analogy to single-item auctions ceases to hold for non-identical actions:  we prove that there is no \emph{generalized Border's theorem}, in the sense of \citet{Gopalan15}, for persuasion with independent actions. Specifically, we show that it is \#P-hard to exactly compute the expected sender utility for the optimal scheme, ruling out Border's-theorem-like approaches to this problem unless the polynomial hierarchy collapses. Our proof starts from the ideas of \cite{Gopalan15}, but our reduction is much more involved and goes through the membership problem for an implicit polytope which encodes a \#P-hard problem --- we elaborate on these differences in Appendix~\ref{app:intro:related_bmd}.  We note that whereas we do rule out computing an explicit representation of the optimal scheme which permits evaluating optimal sender utility, we do not rule out other approaches which might sample the optimal scheme ``on the fly'' in the style of Myerson's optimal auction \cite{Myerson81}--- we leave the intriguing question of whether this is possible as an open problem.

Finally, we consider the black-box model with general distributions, and prove essentially-matching positive and negative results. For our positive result, we exhibit fully polynomial-time approximation scheme (FPTAS) with a bicriteria guarantee. Specifically, our scheme loses an additive $\eps$ in both expected sender utility and incentive-compatibility (as defined in Section \ref{sec:prelim}),  and runs in time polynomial in the number of actions and $\frac{1}{\eps}$. Our negative results show that this is essentially the best possible for information-theoretic reasons: any polynomial-time scheme in the black box model which comes close to optimality must significantly sacrifice incentive compatibility, and vice versa. We note that our scheme is related to some prior work on BMD with black-box distributions \cite{Cai12b, Weinberg14}, but is significantly simpler and more efficient: instead of using the ellipsoid method to optimize over ``reduced forms'', our scheme simply solves a single linear program on a sample from the prior distribution on states of nature. Such simplicity is possible in our setting due to the different notion of incentive compatibility in persuasion, which reduces to incentive compatibility on the sample using the principle of deferred decisions. We elaborate on this connection in Appendix~\ref{app:intro:related_bmd}.

We remark that our results suggest that the differences between persuasion and auction design serve as a double-edged sword. This is evidenced by our negative result for independent model and our ``simple''  positive result for the black-box model.



\subsection{Additional Discussion of Related Work}
\label{sec:intro:related_models}
To our knowledge, \citet{Brocas07} were the first to explicitly consider persuasion through information control. They consider a sender with the ability to costlessly acquire information regarding the payoffs of the receiver's actions, with the stipulation that acquired information is available to both players.  This is technically equivalent to our (and Kamenica and Gentzkow's~\cite{Kamenica11})  informed sender who commits to a signaling scheme. Brocas and Carrillo restrict attention to a particular setting with  two states of nature and three actions,  and characterize optimal policies for the sender and their associated payoffs. %
Kamenica and Gentzkow's \cite{Kamenica11} Bayesian Persuasion model naturally generalizes \cite{Brocas07} to finite (or infinite yet compact) states of nature and action spaces. They establish a number of properties of optimal information structures in this model; most notably, they characterize settings in which signaling strictly benefits the sender in terms of the convexity/concavity of the sender's payoff as a function of the receiver's posterior belief. 
%
%
%

Since \cite{Brocas07} and \cite{Kamenica11}, an explosion of interest in persuasion problems followed. The basic Bayesian persuasion model underlies, or is closely related to, recent work in a number of different domains: price discrimination by \citet{Bergemann14}, advertising by \citet{Chakraborty12},   security games by \citet{Xu15} and \citet{Rabinovich15},  multi-armed bandits by \citet{crowds} and \citet{bandit_exploration},
medical research by \citet{kolotilin_experimental}, and financial regulation by \citet{stresstest1} and \citet{stresstest2}. Generalizations and variants of the Bayesian persuasion model have also been considered:  \citet{Gentzkow11} consider  multiple senders,  \citet{Alonso14} consider  multiple receivers in a voting setting,  \citet{Gentzkow14}  consider costly information acquisition,   \citet{Rayo10} consider an outside option for the receiver, and \citet{kolotilin_persuasion_private} considers a receiver with private side information. 

Optimal persuasion is a special case of \emph{information structure design} in games, also known as \emph{signaling}. The space of information structures, and their induced equilibria, are characterized by \citet{Bergemann14Bayes}.   Recent work in the CS community has also examined the design of information structures algorithmically. Work by \citet{Emek12}, \citet{Miltersen12}, \citet{Guo13}, and \citet{Dughmi14b}, examine optimal signaling in a variety of auction settings, and presents polynomial-time algorithms and hardness results.  \citet{Dughmi14a} exhibits hardness results for signaling in two-player zero-sum games, and \citet{mixture_selection} present an algorithmic framework and apply it to a number of different signaling problems.

Also related to the Bayesian persuasion model is the extensive literature on \emph{cheap talk} starting with  \citet{Crawford82}. Cheap talk can be viewed as the analogue of persuasion when the sender cannot commit to an information revelation policy. Nevertheless, the commitment assumption in persuasion has been justified on the grounds that it arises organically in repeated cheap talk interactions with a long horizon --- in particular when the sender must balance his short term payoffs with long-term credibility. We refer the reader to the discussion of this phenomenon in \cite{Rayo10}. Also to this point, \citet{Kamenica11} mention that an earlier model of repeated 2-player games with asymmetric information by \citet{Aumann95} is mathematically analogous to Bayesian persuasion.

Various recent models on \emph{selling information} in \cite{Babaioff12,Bergemann15a,Bergemann15b} are quite similar to Bayesian persuasion, with the main difference being that the sender's utility function is replaced with revenue. Whereas \citet{Babaioff12} consider the algorithmic question of selling information when states of nature are explicitly given as input, the analogous algorithmic questions to ours have not been considered in their model. We speculate that some of our algorithmic techniques might be applicable to models for selling information when the prior distribution on states of nature is represented succinctly.

As discussed previously, our results involve exercising and generalizing ideas from prior work in Bayesian mechanism design. We view drawing these connections as one of the contributions of our paper. In Appendix~\ref{app:intro:related_bmd}, we discuss these connections and differences at length.

\section{Preliminaries}
\label{sec:prelim}
\def\emd{\qopname\relax n{EMD}}

In a persuasion game, there are two players: a \emph{sender} and a  \emph{receiver}. The receiver is faced with selecting an action from $[n] = \set{1,\ldots,n}$, with an a-priori-unknown payoff to each of the sender and receiver. We assume payoffs  are a function of an unknown \emph{state of nature} $\theta$, drawn from an abstract set $\Theta$ of potential realizations of nature. Specifically, the sender and receiver's payoffs are functions $s,r : \Theta \cross [n] \to \RR$, respectively. We use $\bvec{r}=\bvec{r}(\theta) \in \RR^n$ to denote the receiver's payoff vector as a function of the state of nature, where $r_i(\theta)$ is the receiver's payoff if he takes action $i$ and the state of nature is  $\theta$. Similarly $\bvec{s}=\bvec{s}(\theta) \in \RR^n$ denotes the sender's payoff vector, and $s_i(\theta)$ is the sender's payoff if the receiver takes action $i$ and the state is $\theta$. Without loss of generality, we often conflate the abstract set $\Theta$ indexing states of nature with the set of realizable payoff vector pairs $(\bvec{s},\bvec{r})$ --- i.e., we think of  $\Theta$ as a subset of  $\RR^n \cross \RR^n$. 
We assume that $\Theta$ is finite for notational convenience, though this is not needed for our results in Section \ref{sec:general}.


In Bayesian persuasion, it is assumed that the state of nature is  a-priori unknown to the receiver, and drawn from a common-knowledge prior distribution $\lambda$ supported on $\Theta$. The sender, on the other hand, has access to the realization of $\theta$, and can commit to a policy of partially revealing information regarding its realization before the receiver selects his action. Specifically, the sender commits to a \emph{signaling scheme} $\varphi$, mapping (possibly randomly) states of nature $\Theta$ to a family of \emph{signals} $\Sigma$. For $\theta \in \Theta$, we use $\varphi(\theta)$ to denote the (possibly random) signal selected when the state of nature is $\theta$. Moreover, we use $\varphi(\theta,\sigma)$ to denote the probability of selecting the signal $\sigma$ given a state of nature $\theta$. An algorithm \emph{implements} a signaling scheme $\varphi$ if it takes as input a state of nature $\theta$, and samples the random variable $\varphi(\theta)$.

Given a signaling scheme $\varphi$ with signals $\Sigma$, each signal $\sigma \in \Sigma$ is realized with probability $\alpha_\sigma = \sum_{\theta \in \Theta} \lambda_\theta \varphi(\theta,\sigma)$. Conditioned on the signal $\sigma$, the expected payoffs to the receiver of the various actions are summarized by the vector $\bvec{r}(\sigma) =\frac{1}{\alpha_\sigma} \sum_{\theta \in \Theta} \lambda_\theta \varphi(\theta,\sigma) \bvec{r}(\theta)$. Similarly, the sender's payoff as a function of the receiver's action are summarized by $\bvec{s}(\sigma) =\frac{1}{\alpha_\sigma} \sum_{\theta \in \Theta} \lambda_\theta \varphi(\theta,\sigma) \bvec{s}(\theta)$. 
On receiving a signal $\sigma$, the receiver  performs a Bayesian update and selects an action $i^*(\sigma) \in \argmax_i r_i(\sigma)$ with expected receiver utility $\max_i r_i(\sigma)$. This induces utility $s_{i^*(\sigma)}(\sigma)$ for the sender.   In the event of ties when selecting $i^*(\sigma)$, we assume those ties are broken in favor of the sender.

We adopt the perspective of a sender looking to design $\varphi$  to maximize her expected utility $\sum_{\sigma} \alpha_\sigma s_{i^*(\sigma)}(\sigma)$, in which case we say $\varphi$ is  \emph{optimal}.  When $\varphi$  yields expected sender utility within an additive [multiplicative] $\eps$ of the best possible, we say it is \emph{$\eps$-optimal} [$\eps$-approximate]  in the additive [multiplicative] sense.  
A simple revelation-principle style argument \cite{Kamenica11} shows that an optimal signaling scheme need not use more than $n$ signals, with one \emph{recommending} each action. Such a \emph{direct} scheme $\varphi$ has signals  $\Sigma=\set{\sigma_1, \ldots, \sigma_n}$, and satisfies $r_i(\sigma_i) \geq r_j(\sigma_i)$ for all $i,j \in [n]$. 
We think of $\sigma_i$ as a signal recommending action $i$, and the requirement $r_i(\sigma_i) \geq \max_j r_j(\sigma_i)$ as an \emph{incentive-compatibility (IC)} constraint on our signaling scheme. We can now write the sender's optimization problem as the following LP with variables $\set{\varphi(\theta,\sigma_i) : \theta \in \Theta, i \in [n] }$.
\begin{lp}\label{lp:persuasion}
	\maxi{\sum_{\theta \in \Theta} \sum_{i=1}^n \lambda_\theta \varphi(\theta,\sigma_i) s_i(\theta)}
	\st
	\qcon{\sum_{i=1}^n \varphi(\theta,\sigma_i) = 1}{\theta \in \Theta}
	\qcon{\sum_{\theta \in \Theta} \lambda_\theta \varphi(\theta,\sigma_i) r_i(\theta) \geq \sum_{\theta \in \Theta} \lambda_\theta \varphi(\theta,\sigma_i) r_j(\theta)}{i,j \in [n]}
	\qcon{\varphi(\theta,\sigma_i) \geq 0}{\theta \in \Theta, i \in [n]}
\end{lp}%
%
%
For our results in Section \ref{sec:general}, we relax our incentive constraints by assuming that the receiver follows the recommendation  so long as it approximately maximizes his utility --- for a parameter $\eps > 0$, we relax our requirement to  $r_i(\sigma_i) \geq \max_j r_j(\sigma_i) - \eps$, which translates to the relaxed IC constraints    $\sum_{\theta \in \Theta} \lambda_\theta \varphi(\theta,\sigma_i) r_i(\theta) \geq \sum_{\theta \in \Theta} \lambda_\theta \varphi(\theta,\sigma_i) (r_j(\theta)- \eps)$ in LP \eqref{lp:persuasion}. We call such  schemes \emph{$\eps$-incentive compatible ($\eps$-IC)}. We judge the suboptimality of an $\eps$-IC scheme relative to the best (absolutely) IC scheme; i.e., in a bi-criteria sense.

Finally, we note that expected utilities, incentive compatibility, and optimality are properties not only of a signaling scheme $\varphi$, but also of the distribution $\lambda$ over its inputs. When  $\lambda$ is not clear from context and $\varphi$ is supported on a superset of  $\lambda$, we often say that a signaling scheme $\varphi$ is IC [$\eps$-IC] for $\lambda$, or optimal [$\eps$-optimal] for $\lambda$. We also use $u_{s}(\varphi,\lambda)$ to denote the expected sender utility $\sum_{\theta \in \Theta} \sum_{i=1}^n \lambda_\theta \varphi(\theta,\sigma_i) s_i(\theta)$.

\section{Persuasion with I.I.D. Actions}
\label{sec:iid}
In this section, we assume the payoffs of different actions are independently and identically distributed (i.i.d.) according to an explicitly-described marginal distribution. Specifically,  each state of nature $\theta$ is a vector in $\Theta=[m]^n$ for a parameter $m$, where $\theta_i \in [m]$ is the \emph{type} of action $i$. Associated with each type $j\in [m]$ is a pair $(\xi_j,\rho_j) \in \RR^2$, where $\xi_j$ [$\rho_j$] is the payoff to the sender [receiver] when the receiver chooses an action with type $j$.  We are given a marginal distribution over types, described by a vector $\bvec{q}=(q_1,...,q_m) \in \Delta_m$. We assume each action's type is drawn independently according to $\bvec{q}$; specifically, the prior distribution $\lambda$ on states of nature is given by $\lambda(\theta) = \prod_{i\in[n]}q_{\theta_{i}}$. For convenience, we let $\bvec{\xi}=(\xi_{1},...,\xi_{m}) \in \mathbb{R}^m$ and $\bvec{\rho}=(\rho_{1},...,\rho_{m}) \in \mathbb{R}^m$ denote the type-indexed vectors of sender and receiver payoffs, respectively.    We assume $\bvec{\xi}$, $\bvec{\rho}$, and $\bvec{q}$ --- the parameters describing an i.i.d. persuasion instance --- are given explicitly. 

Note that the number of states of nature is $m^n$, and therefore the natural representation of a signaling scheme has $n m^n$ variables. Moreover, the  natural linear program for the persuasion problem in Section~\ref{sec:prelim} has an exponential in $n$ number of both variables and constraints. Nevertheless, as mentioned in Section \ref{sec:prelim} we seek only to implement an optimal or near-optimal scheme $\varphi$ as an oracle which takes as input $\theta$ and samples a signal $\sigma \sim \varphi(\theta)$. Our algorithms  will run in time polynomial in $n$ and $m$, and will optimize over a space of succinct  ``reduced forms'' for signaling schemes which we term \emph{signatures}, to be described next.

For a state of nature $\theta$, define the matrix $M^\theta \in \set{0,1}^{n \times m}$ so that $M^\theta_{ij} =1$ if and only if action $i$ has type $j$ in $\theta$ (i.e. $\theta_i=j$). Given an i.i.d~prior $\lambda$ and  a signaling scheme $\varphi$  with signals $\Sigma =\set{\sigma_1,\ldots,\sigma_n}$, for each $i \in [n]$ let $\alpha_i=  \sum_{\theta}\lambda(\theta)\varphi(\theta,\sigma_i)$ denote the probability of sending $\sigma_i$, and let $M^{\sigma_i} =   \sum_{\theta}\lambda(\theta)\varphi(\theta,\sigma_i)M^{\theta}$.  Note that $M^{\sigma_i}_{jk}$ is the joint probability that action $j$ has type $k$ and the scheme outputs $\sigma_i$. Also note that each row of $M^{\sigma_i}$ sums to $\alpha_i$, and the $j$th row represents the un-normalized posterior type distribution of action $j$  given signal $\sigma_i$. We call $\mathcal{M}=(M^{\sigma_1},...,M^{\sigma_n}) \in \mathbb{R}^{n\times m \times n}$ the {\it signature} of  $\varphi$. The sender's objective and receiver's IC constraints can both be expressed in terms of the signature. In particular, using $M_j$ to denote the $j$th row of a matrix $M$, the IC constraints are $\bvec{\rho} \cdot M^{\sigma_i}_{i} \geq \bvec{\rho} \cdot M^{\sigma_i}_{j}$ for all $i,j\in[n]$, and the sender's expected utility assuming the receiver follows the scheme's recommendations is $\sum_{i\in[n]}\bvec{\xi} \cdot M^{\sigma_i}_{i}$.   

We say  $\mathcal{M}=(M^{\sigma_1},...,M^{\sigma_n}) \in \mathbb{R}^{n\times m \times n}$ is  \emph{realizable} if there exists a signaling scheme $\varphi$ with $\mathcal{M}$ as its signature. Realizable signatures constitutes a polytope $\mathcal{P} \subseteq \mathbb{R}^{n\times m \times n}$, which has an exponential-sized extended formulation as shown Figure \ref{fig:extended}. 
Given this characterization, the sender's optimization problem can be written as a linear program in the space of signatures, shown in Figure \ref{fig:optiid}:


\begin{figure}
	\centering
	\begin{minipage}{.5\textwidth}
		\centering
		\begin{lp*}
			\qcon{M^{\sigma_i}=\sum_{\theta}\lambda(\theta)\varphi(\theta,\sigma_i)M^{\theta} }{i=1,\ldots,n}
			\qcon{\sum_{i =1}^{n} \varphi(\theta,\sigma_i) = 1 }{\theta \in \Theta}
			\qcon{\varphi(\theta,\sigma_i) \geq 0}{\theta \in \Theta, i \in [n]}
		\end{lp*}%
		\vspace{-14pt}
		\captionof{figure}{Realizable Signatures $\P$}
		\label{fig:extended}
	\end{minipage}%
	\begin{minipage}{.5\textwidth}
		\centering
		\begin{lp*}
			\maxis{\sum_{i = 1}^{n} \bvec{\xi} \cdot M_{i}^{\sigma_i} }   
			\sts
			\qcon{ \bvec{\rho} \cdot M_{i}^{\sigma_i}\geq  \bvec{\rho} \cdot M_{j}^{\sigma_i} }{ i,j \in [n]}
			\con{ (M^{\sigma_1},...,M^{\sigma_n}) \in \mathcal{P} }%
		\end{lp*}
		\vspace{-14pt}
		\captionof{figure}{Persuasion in Signature Space}
		\label{fig:optiid}
	\end{minipage}%
\end{figure}

\subsection{Symmetry of the Optimal Signaling Scheme}
\label{sec:iid:symmetry}

We now show that there always exists a ``symmetric" optimal  scheme when actions are i.i.d. Given a signature $\mathcal{M}=(M^{\sigma_1},...,M^{\sigma_n})$, it will  sometimes be convenient to  think of it as the set of pairs $\{(M^{\sigma_i},\sigma_i)\}_{i\in[n]}$. 

\begin{definition}\label{def:symmetry}
	A signaling scheme $\varphi$ with signature     
	$\{ (M^{\sigma_i},\sigma_i)\}_{i \in [n]}$  is \emph{symmetric} if there exist $\bvec{x},\bvec{y}\in \RR^{m}$ 
	such that $M^{\sigma_i}_{i}=\bvec{x}$ for all $i\in [n]$ and $M^{\sigma_i}_{j}=\bvec{y}$ for all $j \not = i$. The pair $(\bvec{x},\bvec{y})$ is  the \emph{$s$-signature} of $\varphi$.
\end{definition}
In other words, a symmetric signaling scheme sends each signal with equal probability $||\bvec{x}||_1$, and induces only two different posterior type distributions  for actions: $\frac{\bvec{x}}{||\bvec{x}||_1}$ for the recommended action, and $\frac{\bvec{y}}{||\bvec{y}||_1}$ for  the others. We call $(\bvec{x},\bvec{y})$ \emph{realizable} if there exists a signaling scheme with $(\bvec{x},\bvec{y})$ as its $s$-signature. The family of realizable $s$-signatures  constitutes a polytope $\mathcal{P}_s$, and has an extended formulation by adding the variables $\bvec{x},\bvec{y} \in \RR^m$ and constraints $M^{\sigma_i}_i = \bvec{x}$ and $M^{\sigma_i}_j = \bvec{y}$ for all $i,j \in [n]$ with $i \neq j$ to the extended formulation of (asymmetric) realizable signatures from Figure \ref{fig:extended}.

We make two simple observations regarding realizable $s$-signatures. First,   $||\bvec{x}||_{1}=||\bvec{y}||_{1}=\frac{1}{n}$ for each $(\bvec{x},\bvec{y})\in \mathcal{P}_s$, and this is because both $||\bvec{x}||_1$ and $||\bvec{y}||_1$ equal the probability of each of the $n$ signals. Second, since the signature must be consistent with prior marginal distribution $\bvec{q}$,  we have  $\bvec{x} + (n-1) \bvec{y} = \sum_{i=1}^n M^{\sigma_i}_{1}=\bvec{q}$. We show that restricting to symmetric signaling schemes is without loss of generality.
\begin{theorem}
	\label{thm:symmetry}When the action payoffs are i.i.d., there exists an optimal and incentive-compatible signaling
	scheme which is symmetric.
\end{theorem} 
Theorem~\ref{thm:symmetry} is proved in Appendix~\ref{app:iid:symmetry}. At a high level, we show that optimal signaling schemes are closed with respect to two operations:  \emph{convex combination} and \emph{permutation}. Specifically, a convex combination of realizable signatures --- viewed as vectors in $\RR^{n \times m \times n}$ --- is realized by the corresponding ``random mixture'' of signaling schemes, and this operation preserves optimality. The proof of this fact follows easily from the fact  that linear program in Figure \ref{fig:optiid} has a convex family of optimal solutions. Moreover, given a permutation $\pi \in \SS_n$ and an optimal signature $\M = \{(M^{\sigma_i},\sigma_i)\}_{i\in[n]}$ realized by signaling scheme $\varphi$, the ``permuted'' signature $\pi ( \M) = \{(\pi M^{\sigma_i},\sigma_{\pi(i)})\}_{i\in[n]}$ --- where premultiplication of a matrix by $\pi$ denotes permuting the rows of the matrix ---  is realized by the ``permuted'' scheme $\varphi_{\pi}(\theta) = \pi (\varphi(\pi^{-1} (\theta)))$, which is also optimal. The proof of this fact follows from the ``symmetry'' of the (i.i.d.) prior distribution about the different actions.  Theorem~\ref{thm:symmetry}  is then proved constructively as follows: given a realizable optimal signature $\M$, the ``symmetrized'' signature $\bar{\M}=\frac{1}{n!} \sum_{\pi \in \SS_n} \pi(\M)$ is realizable, optimal, and symmetric.

\def\high{\qopname\relax n{\mathbf{HIGH}}}
\def\low{\qopname\relax n{\mathbf{LOW}}}

\subsection{Implementing the Optimal Signaling Scheme}
\label{sec:iid:optimal}
We now exhibit a polynomial-time algorithm for persuasion in the i.i.d.~model. Theorem \ref{thm:symmetry} 
permits re-writing the optimization problem in Figure \ref{fig:optiid} as follows, with variables $\bvec{x},\bvec{y}\in \mathbb{R}^{m}$.
\begin{lp}\label{lp:iid:opt}	
	\maxi{n\bvec{\xi} \cdot \bvec{x} }   
	\st
	\con{ \bvec{\rho} \cdot \bvec{x} \geq \bvec{\rho} \cdot \bvec{y} }
	\con{ (\bvec{x},\bvec{y})\in\mathcal{P}_{s} }
\end{lp}%
Problem \eqref{lp:iid:opt}  cannot be solved directly, since $\P_s$ is defined by an extended formulation with exponentially many variables and constraints, as described in Section~\ref{sec:iid:symmetry}.  Nevertheless, we make use of a 
connection between symmetric signaling schemes and single-item
auctions with i.i.d.~bidders to solve \eqref{lp:iid:opt} using the Ellipsoid method. Specifically, we show  a one-to-one correspondence between symmetric signatures  and (a subset of)  symmetric reduced forms of single-item auctions with i.i.d.~bidders,  defined as follows.
\begin{definition}[\cite{Border91}]
	Consider a single-item auction setting with $n$ i.i.d. bidders and $m$ types for each bidder, where each bidder's type is distributed according to $\bvec{q} \in \Delta_m$. An  \emph{allocation rule} is a randomized function $A$ mapping a type profile $\theta \in [m]^n$ to a winner $A(\theta) \in [n] \union \set{*}$, where $*$ denotes not allocating the item. 
	We say the allocation rule has \emph{symmetric reduced form}  $\bvec{\tau}\in [0,1]^m$ if for each bidder $i \in [n]$ and type $j \in [m]$,  $\tau_j$ is the conditional probability of $i$ receiving the item given she has type $j$.
\end{definition}
\noindent When $\bvec{q}$ is clear from context, we say $\bvec{\tau}$ is \emph{realizable} if there exists an allocation rule with $\bvec{\tau}$ as its symmetric reduced form. We say an algorithm \emph{implements} an allocation rule $A$ if it takes as input $\theta$, and samples $A(\theta)$. 
\begin{theorem}
	\label{thm:i.i.d. opt} Consider the Bayesian Persuasion problem with $n$ i.i.d. actions and $m$ types, with parameters $\bvec{q} \in \Delta_m$, $\bvec{\xi} \in \RR^m$, and $\bvec{\rho} \in \RR^m$ given explicitly. An optimal and incentive-compatible signaling scheme can be implemented in $\poly(m,n)$ time.
\end{theorem}
Theorem \ref{thm:i.i.d. opt} is a consequence of the following set of lemmas.
\begin{lemma}\label{lem:reduced_form}
	Let $(\bvec{x},\bvec{y}) \in [0,1]^m \cross [0,1]^m$, and define $\bvec{\tau}=(\frac{x_1}{q_1},...,\frac{x_m}{q_m})$.  The pair $(\bvec{x},\bvec{y})$ is a realizable $s$-signature if and only if  (a) $||\bvec{x}||_1 = \frac{1}{n}$,  (b) $\bvec{x}+(n-1)\bvec{y}=\bvec{q}$, and (c) $\bvec{\tau}$ is a realizable symmetric reduced form of an allocation rule with $n$ i.i.d. bidders, $m$ types, and type distribution $\bvec{q}$. Moreover, assuming $\bvec{x}$ and $\bvec{y}$ satisfy (a), (b) and (c), and given black-box access to an allocation rule $A$ with symmetric reduced form $\bvec{\tau}$, a signaling scheme with $s$-signature $(\bvec{x},\bvec{y})$ can be implemented in $\poly(n,m)$ time.
\end{lemma}

\begin{lemma}\label{lem:iid_solve}
	An optimal realizable $s$-signature, as described by LP  \eqref{lp:iid:opt}, is computable in $\poly(n,m)$ time. 
\end{lemma}
%
%
%
%

\begin{lemma} (See \cite{Cai12a,Saeed12}) \label{lem:iid_implement}
	Consider a single-item auction setting with $n$ i.i.d. bidders and $m$ types for each bidder, where each bidder's type is distributed according to $\bvec{q} \in \Delta_m$.  Given a realizable symmetric reduced form $\bvec{\tau} \in [0,1]^m$, an allocation rule with reduced form $\bvec{\tau}$ can be implemented in $\poly(n,m)$ time.
\end{lemma}

The proofs of Lemmas \ref{lem:reduced_form} and \ref{lem:iid_solve} can be found in Appendix \ref{app:iid:opt}. 
The proof of Lemma~\ref{lem:reduced_form} builds  a correspondence between $s$-signatures of signaling schemes and certain reduced-form allocation rules. Specifically, actions correspond to bidders,  action types correspond to bidder types, and signaling $\sigma_i$ corresponds to assigning the item to bidder $i$. The expression of the reduced form in terms of the s-signature then follows from Bayes' rule. Lemma~\ref{lem:iid_solve} follows from Lemma \ref{lem:reduced_form}, the ellipsoid method, and the fact that symmetric reduced forms admit an efficient separation oracle (see \cite{Border91,Border07,Cai12a,Saeed12}).


\subsection{A Simple $(1-\frac{1}{e})$-Approximate Scheme}
\label{sec:iid:simple}

Our next result is a ``simple'' signaling scheme which obtains a $(1-1/e)$ multiplicative approximation  when payoffs are nonnegative. This algorithm has the distinctive property that it signals \emph{independently} for each action, and therefore implies that approximately optimal persuasion can be parallelized among multiple colluding  senders, each of whom only has access to the type of one or more of the actions. 

Recall from Section~\ref{sec:iid:symmetry} that an s-signature $(\bvec{x},\bvec{y})$ satisfies $||\bvec{x}||_1=||\bvec{y}||_1=\frac{1}{n}$ and $\bvec{x} + (n-1) \bvec{y} = \bvec{q}$. Our simple scheme, shown in Algorithm~\ref{alg:1-1/e}, works with the following explicit linear programming relaxation of optimization problem~\eqref{lp:iid:opt}.
\begin{lp}\label{lp:iid:apx}	
	\maxi{n\bvec{\xi} \cdot \bvec{x} }   
	\st
	\con{ \bvec{\rho} \cdot \bvec{x} \geq \bvec{\rho} \cdot \bvec{y} }
	\con{ \bvec{x} + (n-1) \bvec{y} = \bvec{q} }
	\con{ ||\bvec{x}||_{1}=\frac{1}{n} }
	\con{ \bvec{x},\bvec{y}\geq0 }
\end{lp}%
\begin{alg}
	\INPUT Sender payoff vector $\bvec{\xi}$, receiver payoff vector
	$\bvec{\rho}$, prior distribution $\bvec{q}$
	
	\INPUT State of nature $\theta\in[m]^{n}$
	
	\OUTPUT An $n$-dimensional binary signal $\sigma\in\{\high,\low\}^{n}$
	
	\STATE Compute an optimal solution $(\bvec{x}^{*},\bvec{y}^{*})$ linear program~\eqref{lp:iid:apx}.
	\STATE For each action $i$ independently, set  component signal $o_i$ to $\high$ with probability
	$\frac{x_{\theta_{i}}^{*}}{q_{\theta_{i}}}$ and to $\low$ otherwise, where $\theta_i$ is the type of action $i$ in the input state $\theta$.
	\STATE Return $\sigma=(o_{1},...,o_{n})$. 		
	\caption{Independent Signaling Scheme\label{alg:1-1/e}}
\end{alg}

Algorithm~\ref{alg:1-1/e} has a simple and instructive interpretation. It computes the optimal solution $(\bvec{x}^*,\bvec{y}^*)$ to  the relaxed problem \eqref{lp:iid:apx}, and uses this solution as a guide for signaling \emph{independently} for each action. The algorithm selects, independently for each action $i$, a component signal $o_i \in \set{\high,\low}$. In particular, each $o_i$ is chosen so that $\Pr[o_i=\high] = \frac{1}{n}$, and moreover the events $o_i=\high$ and $o_i=\low$ induce the  posterior beliefs $n \bvec{x}^*$ and $n \bvec{y}^*$, respectively, regarding the type of action $i$.

The signaling scheme implemented by Algorithm~\ref{alg:1-1/e} approximately matches the optimal value of~\eqref{lp:iid:apx}, as shown in Theorem \ref{thm:iid_apx}, assuming the receiver is rational and therefore selects an action with a $\high$ component signal if one exists. We note that the scheme of Algorithm~\ref{alg:1-1/e}, while not a direct scheme as described, can easily be converted into one; specifically, by recommending an action whose component signal is  $\high$ when one exists (breaking ties arbitrarily), and recommending an arbitrary action otherwise.
Theorem \ref{thm:iid_apx} follows from the fact that $(\bvec{x}^*,\bvec{y}^*)$ is an optimal solution to LP \eqref{lp:iid:apx}, the fact that the posterior type distribution of an action $i$ is $n\bvec{x}^*$ when $o_i=\high$ and $n \bvec{y}^*$ when $o_i=\low$, and the fact that each component signal is high independently with probability $\frac{1}{n}$. We defer the formal proof to Appendix \ref{app:iid:apx}.
\begin{theorem}\label{thm:iid_apx}
	Algorithm~\ref{alg:1-1/e} runs in $poly(m,n)$ time, and serves as a $(1-\frac{1}{e})$-approximate signaling scheme for the Bayesian Persuasion problem with $n$ i.i.d. actions, $m$ types, and nonnegative payoffs.
\end{theorem}


\begin{remark}
	Algorithm \ref{alg:1-1/e} signals independently for each action. This conveys an interesting conceptual message. That
	is, even though the optimal signaling scheme might induce posterior beliefs which correlate different actions, it is nevertheless true that signaling for each action independently yields an approximately optimal signaling scheme. As a consequence, collaborative persuasion by multiple parties (the senders), each of whom observes the payoff of one or more actions, is a task that can be parallelized, requiring no coordination when actions are identical and independent and only an approximate solution is sought.  We leave open the question of whether this is possible when action payoffs are independently but not identically distributed.
\end{remark}


\section{Complexity Barriers to Persuasion with Independent Actions}
\label{sec:indep}

In this section, we consider optimal persuasion with independent   action payoffs as in Section \ref{sec:iid}, albeit with action-specific marginal distributions given explicitly. Specifically, for each action $i$ we are given a distribution $q^i \in \Delta_{m_i}$ on $m_i$ types, and each type $j \in [m_i]$ of action $i$ is associated with a sender payoff $\xi^i_j \in \RR$ and a receiver payoff $\rho^i_j \in \RR$. The positive results of Section \ref{sec:iid} draw a connection between optimal persuasion in the special case of identically distributed actions and Border's characterization of reduced-form single-item auctions with i.i.d.~bidders. One might expect this connection to generalize to the independent non-identical persuasion setting, since Border's theorem extends to single-item auctions with independent non-identical bidders. Surprisingly, we show that this analogy to Border's characterization fails to generalize. We prove the following theorem. 
\begin{theorem}\label{thm:hardness}
	Consider the Bayesian Persuasion problem with independent actions, with action-specific payoff distributions given explicitly. It is $\#P$-hard to compute the optimal expected sender utility.
\end{theorem}
Invoking the framework of \citet{Gopalan15}, this rules out a \emph{generalized Border's theorem} for our setting, in the sense defined by~\cite{Gopalan15}, unless the polynomial hierarchy collapses to $P^{NP}$. We view this result as illustrating some of the important differences between persuasion and mechanism design. 

The proof of Theorem \ref{thm:hardness} is rather involved. We defer the full proof to Appendix \ref{app:indep}, and only present a sketch here.  Our proof starts from the ideas of  \citet{Gopalan15}, who show the \#P-hardness for revenue or welfare maximization in several mechanism design problems. In one case, \cite{Gopalan15} reduce from the $\#P$-hard problem of computing the \emph{Khintchine constant} of a vector. Our reduction also starts from this problem, but is much more involved:\footnote{In \cite{Gopalan15}, Myerson's characterization is used to show that optimal mechanism design in a public project setting directly encodes computation of the Khintchine constant. No analogous direct connection seems to hold here.} First, we exhibit a polytope which we term the \emph{Khintchine polytope}, and show that computing the Khintchine constant reduces to linear optimization over the Khintchine polytope. Second, we present a reduction from the membership problem for the Khintchine polytope to the computation of optimal sender utility in a particularly-crafted instance of persuasion with independent actions. Invoking the polynomial-time equivalence between membership checking and optimization (see, e.g., \cite{Groetschel88}), we conclude the \#P-hardness of our problem. The main technical challenge we overcome is in the second step of our proof: given a point $x$ which may or may not be in the Khintchine polytope $\K$, we construct a persuasion instance and a threshold $T$ so that points in $\K$ encode signaling schemes, and the optimal sender utility is at least $T$ if and only if $x \in \K$ and the scheme corresponding to $x$ results in sender utility $T$.

\subsection*{Proof Sketch of Theorem~\ref{thm:hardness}}
The \emph{Khintchine problem},  shown to be \#P-hard in \cite{Gopalan15}, is to compute the \emph{Khintchine constant} $K(a)$ of a given vector $a \in \RR^n$, defined as $K(a) =  \Ex_{\theta \sim \{\pm 1\}^n} [ |\theta \cdot a|]$ where  $\theta$ is drawn uniformly at random from $\{\pm 1\}^n$. To relate the Khintchine problem to Bayesian persuasion, we begin with a persuasion instance with $n$ \emph{i.i.d.}~actions and two action types, 
which we refer to as \emph{type -1} and \emph{type +1}. The state of nature is a uniform random draw from the set $\{\pm 1 \}^n$, with the $i$th entry specifying the type of action~$i$. 
We call this instance the \emph{Khintchine-like} persuasion setting.  As in Section \ref{sec:iid}, we still use the \emph{signature} to capture the payoff-relevant features of a signaling scheme, but we pay special attention to signaling schemes which use only \emph{two} signals, in which case we represent them using a \emph{two-signal signature} of the form $(M^1,M^2) \in \RR^{n \times 2} \cross \RR^{n \times 2}$. 
The \emph{Khintchine polytope} $\mathcal{K}(n)$ is then defined as the (convex) family of all \emph{realizable} two-signal signatures for the Khintchine-like persuasion problem with an additional constraint: each signal is sent  with probability exactly $\frac{1}{2}$. We first prove that general linear optimization over $\K(n)$ is \#P-hard by encoding computation of the Khintchine constant as linear optimization over $\K(n)$. In this reduction, the optimal solution in $\K(n)$ is the signature of  the two-signal scheme $\varphi(\theta) = sign(\theta \cdot a)$, which signals $+$ and $-$ each with  probability $\frac{1}{2}$.


To reduce the membership problem for the Khintchine polytope to optimal Bayesian persuasion, the main challenges come from our restrictions on $\K(n)$, namely to schemes with two signals which are equally probable. 
Our reduction incorporates three key ideas. The \emph{first} is to design a persuasion instance in which the optimal signaling scheme uses only two signals. The instance we define will have $n+1$ actions. Action $0$ is \emph{special} -- it deterministically results in sender utility $\epsilon>0$ (small enough) and receiver utility $0$. 
The other $n$ actions are \emph{regular}. Action $i>0$ \emph{independently} results in sender utility $- a_i$ and receiver utility $a_i$ with probability $\frac{1}{2}$ (call this type $1_i$), or sender utility $-  b_i$ and receiver utility $b_i$ with probability $\frac{1}{2}$ (call this type $2_i$), for  $a_i$ and $b_i$ to be set later.  Note that the sender and receiver utilities are \emph{zero-sum} for both types.  
Since the special action is deterministic and the probability of its (only) type is $1$ in any signal, we  can interpret any $(M^1,M^2)\in \mathcal{K}(n)$ as a two-signal signature for our persuasion instance (the row corresponding to the special action $0$ is implied). 
We show that restricting to two-signal schemes is without loss of generality in this persuasion instance. The proof tracks the following intuition: due to the zero-sum nature of regular actions, any additional information regarding  regular actions would benefit the receiver and harm the sender. Consequently, sender does not reveal any information which distinguishes between different regular actions. Formally, we prove that there always exists an optimal signaling scheme with only two signals: one signal recommends the special action, and the other recommends some regular action. 

We denote the  signal that recommends the special action $0$ by $\sigma_+$ (indicating that the sender derives positive utility $\epsilon$), and denote the other signal by $\sigma_-$ (indicating that the sender derives negative utility, as we show). The \emph{second} key idea concerns choosing appropriate values for $\{a_i\}_{i=1}^{n},\{b_i\}_{i=1}^n$ for a given two-signature $(M^1,M^2)$ to be tested. We choose these values to satisfy the following two properties: (1) For all regular actions, the signaling scheme  implementing $(M^1,M^2)$ (if it exists) results in the same sender utility $-1$  (thus receiver utility $1$)  conditioned on $\sigma_-$ and the same sender utility $0$  conditioned on $\sigma_+$; (2)  the \emph{maximum possible} expected sender utility from $\sigma_-$, i.e., the sender utility conditioned on $\sigma_-$ multiplied by the probability of $\sigma_-$, is $-\frac{1}{2}$.  As a result of Property (1), if $(M^1,M^2) \in \K (n)$ then the corresponding signaling scheme $\varphi$ is IC and results in expected sender utility $T=\frac{1}{2}\epsilon - \frac{1}{2}$  (since each signal is sent with probability $\frac{1}{2}$). Property (2) implies that $\varphi$ results in the maximum possible expected sender utility from $\sigma_-$.

We now run into a challenge: the existence of a signaling scheme with expected sender utility  $T=\frac{1}{2}\epsilon - \frac{1}{2}$ does not necessarily imply that $(M^1,M^2) \in \K(n)$ if $\eps$ is large. Our \emph{third} key idea is to set $\eps>0$  ``sufficiently small''  so that any optimal signaling scheme must result in  the maximum possible expected sender utility  $-\frac{1}{2}$ from signal $\sigma_-$ (see Property (2) above). In other words, we must make $\epsilon$  so small so that the sender prefers to not sacrifice \emph{any} of her payoff from $\sigma_-$ in order to gain utility from the special action recommended by $\sigma_+$. We show that such an $\epsilon$ exists with polynomially many bits. We prove its existence by arguing that the polytope of incentive-compatible two-signal signatures has polynomial bit complexity, and therefore an $\eps>0$ that is smaller than the ``bit complexity'' of the vertices would suffice.



As a result of this choice of $\epsilon$, if the optimal sender utility is precisely $T=\frac{1}{2}\epsilon - \frac{1}{2}$ then we know that signal $\sigma_+$ must be sent with probability $\frac{1}{2}$ since the expected sender utility  from signal $\sigma_-$ must be $-\frac{1}{2}$. We show that this, together with the specifically constructed $\{a_i\}_{i=1}^{n},\{b_i\}_{i=1}^n$, is sufficient to guarantee that the optimal signaling scheme must implement  the given two-signature $(M^1,M^2)$, i.e., $(M^1,M^2) \in \K(n)$. When the optimal optimal sender utility is strictly greater than $\frac{1}{2}\epsilon - \frac{1}{2}$, the optimal signaling scheme does not implement $(M^1,M^2)$, but we show that it can be post-processed into one that does. 

\section{The General Persuasion Problem}
\label{sec:general}

We now turn our attention to the Bayesian Persuasion problem when the payoffs of different actions are arbitrarily correlated, and the joint distribution $\lambda$ is presented as  a black-box sampling oracle.  We assume that payoffs are normalized to lie in the bounded interval, and prove essentially matching positive and negative results. Our positive result is a fully polynomial-time approximation scheme for optimal persuasion with a bi-criteria guarantee; specifically, we achieve approximate optimality and approximate incentive compatibility in the additive sense described in Section \ref{sec:prelim}. Our negative results show that such a bi-criteria loss is inevitable in the black box model for information-theoretic reasons.



\subsection{A Bicriteria FPTAS}


\begin{theorem}\label{thm:blackbox}
	Consider the Bayesian Persuasion problem in the black-box oracle model with $n$ actions and payoffs in $[-1,1]$, and let $\eps >0$ be a parameter. An $\eps$-optimal and $\eps$-incentive compatible  signaling scheme can be implemented in $\poly(n,\frac{1}{\eps})$ time.
\end{theorem}

To prove Theorem~\ref{thm:blackbox}, we show that a simple Monte-Carlo algorithm implements an approximately optimal and approximately incentive compatible scheme $\varphi$. Notably, our algorithm does not compute a representation of the entire signaling scheme $\varphi$ as in Section \ref{sec:iid}, but rather merely samples its output $\varphi(\theta)$ on a given input $\theta$. At a high level, when given as input a state of nature $\theta$, our algorithm first takes $K=\poly(n,\frac{1}{\eps})$ samples from the prior distribution $\lambda$ which, intuitively, serve to place the true state of nature $\theta$ in context. Then the algorithm uses a linear program to compute the optimal $\eps$-incentive compatible  scheme $\tilde{\varphi}$ for the empirical distribution of samples augmented with the input $\theta$. Finally, the algorithm signals as  suggested by $\tilde{\varphi}$ for $\theta$. Details are in Algorithm~\ref{alg:blackbox}, which we instantiate with $\eps >0$ and $K = \ceil{ \frac{256 n^2}{\eps^4} \log(\frac{4n}{\eps})}$. 

We note that relaxing incentive compatibility  is necessary for convergence to the optimal sender utility --- we prove this formally in Section \,\ref{sec:blackbox:barriers}. This is why LP \eqref{lp:blackbox:empirical} 
features relaxed incentive compatibility constraints. Instantiating Algorithm\,\ref{alg:blackbox} with $\eps=0$ results in an exactly incentive compatible scheme which could be far from the optimal sender utility for any finite number of samples $K$, as reflected in Lemma \,\ref{lem:blackbox:apx}.

\begin{alg}
	\PARAMETER  $\eps \geq 0$
	\PARAMETER Integer $K \geq 0$
	\INPUT Prior distribution $\lambda$ supported  on $[-1,1]^{2n}$, given by a sampling oracle
	\INPUT State of nature $\theta \in [-1,1]^{2n}$ 
	\OUTPUT Signal $\sigma \in \Sigma$, where $\Sigma= \set{\sigma_1,\ldots,\sigma_n}$.
	\STATE Draw integer $\ell$ uniformly at random from $\set{1,\ldots,K}$, and denote $\theta_{\ell} = \theta$.
	\STATE Sample $\theta_1, \ldots, \theta_{\ell-1}, \theta_{\ell+1}\ldots,\theta_K$ independently from $\lambda$, and let the multiset $\tilde{\lambda} = \set{\theta_1, \ldots, \theta_K}$ denote the empirical distribution  augmented with the input state  $\theta= \theta_\ell$.
	\STATE Solve linear program \eqref{lp:blackbox:empirical} to obtain the signaling scheme $\tilde{\varphi}: \tilde{\lambda}  \to \Delta(\Sigma)$. \label{step:lpempirical}
	\STATE Output a sample from $ \tilde{\varphi}(\theta) = \tilde{\varphi}(\theta_\ell) $.
	\caption{Signaling Scheme for a Black Box Distribution}
	\label{alg:blackbox}
\end{alg}

\begin{figure}
	\centering
	\begin{lp}
		\label{lp:blackbox:empirical}
		\maxi{\sum_{k=1}^K \sum_{i=1}^n \frac{1}{K} \tilde{\varphi}(\theta_k,\sigma_i) s_i(\theta_k)}
		\st
		\qcon{\sum_{i=1}^n \tilde{\varphi}(\theta_k,\sigma_i) = 1}{k \in [K]}
		\qcon{\sum_{k=1}^K \frac{1}{K} \tilde{\varphi}(\theta_k,\sigma_i) r_i(\theta_k) \geq \sum_{k=1}^K \frac{1}{K} \tilde{\varphi}(\theta_k,\sigma_i) (r_j(\theta_k) - \eps)} {i,j \in [n]}
		\qcon{\tilde{\varphi}(\theta_k,\sigma_i) \geq 0}{k \in [K], i \in [n]}
	\end{lp}
	Relaxed Empirical Optimal Signaling Problem
\end{figure}

Theorem \ref{thm:blackbox} follows from three lemmas pertaining to the scheme $\varphi$ implemented by Algorithm~\ref{alg:blackbox}. Approximate incentive compatibility for $\lambda$ (Lemma \ref{lem:blackbox:ic}) follows from the principle of deferred decisions, linearity of expectations, and the fact that $\tilde{\varphi}$ is approximately incentive compatible for the augmented empirical distribution $\tilde{\lambda}$. 
A similar argument, also based on the principal of deferred decisions and linearity of expectations,  
shows that the expected sender utility from our scheme when $\theta \sim \lambda$ equals the expected optimal value of linear program \eqref{lp:blackbox:empirical}, as stated in  Lemma \ref{lem:blackbox:util}. 
Finally, we show in Lemma \ref{lem:blackbox:apx} that the optimal value of LP \eqref{lp:blackbox:empirical} is  close to the optimal sender utility for $\lambda$ with high probability, and hence also in expectation, when $K=\poly(n,\frac{1}{\eps})$ is chosen appropriately; the proof of this fact invokes standard tail bounds as well as structural properties of linear program \eqref{lp:blackbox:empirical}, and exploits the fact that LP \eqref{lp:blackbox:empirical} relaxes the incentive compatibility constraint. We prove all three lemmas in Appendix~\ref{app:general:fptas}. Even though our proof of Lemma \ref{lem:blackbox:apx} is self-contained, we note that it can be shown to follow from  \cite[Theorem 6]{Weinberg14} with some additional work. 

\begin{lemma}\label{lem:blackbox:ic}
	Algorithm \ref{alg:blackbox} implements an $\eps$-incentive compatible  signaling scheme for prior distribution $\lambda$.
\end{lemma}

\begin{lemma}\label{lem:blackbox:util}
	Assume $\theta \sim \lambda$, and assume the receiver follows the recommendations of Algorithm \ref{alg:blackbox}.  The expected  sender utility  equals the expected optimal value of the linear program \eqref{lp:blackbox:empirical} solved in Step \ref{step:lpempirical}. Both expectations are taken over the random input $\theta$ as well as internal randomness and Monte-Carlo sampling performed by the algorithm.
\end{lemma}

\begin{lemma}\label{lem:blackbox:apx}
	Let  $OPT$ denote the expected sender utility induced by the optimal incentive compatible signaling scheme for distribution $\lambda$. When Algorithm \ref{alg:blackbox} is instantiated with $K \geq \frac{256 n^2}{\eps^4} \log(\frac{4n}{\eps})$ and its input  $\theta$ is drawn from $\lambda$,  the expected optimal value of the linear program \eqref{lp:blackbox:empirical} solved in Step \ref{step:lpempirical} is at least $OPT-\eps$.  Expectation is over the random input $\theta$ as well as the Monte-Carlo sampling performed by the algorithm.
\end{lemma}

\subsection{Information-Theoretic Barriers}
\label{sec:blackbox:barriers}

We now show that our bi-criteria FPTAS is close to the best we can hope for: there is no bounded-sample signaling scheme in the black box model which guarantees incentive compatibility and $c$-optimality for any constant $c < 1$, nor is there such an algorithm which guarantees optimality and $c$-incentive compatibility for any $c < \frac{1}{4}$. Formally, we consider algorithms which implement direct signaling schemes. Such an algorithm takes as input a black-box distribution $\lambda$ supported on $[-1,1]^{2n}$ and a state of nature $\theta \in [-1,1]^{2n}$, where $n$ is the number of actions, and outputs a signal $\sigma \in \set{\sigma_1,\ldots,\sigma_n}$ recommending an action. We say such an algorithm is $\eps$-incentive compatible [$\eps$-optimal] if for every distribution $\lambda$ the signaling scheme $\A(\lambda)$ is $\eps$-incentive compatible [$\eps$-optimal] for $\lambda$. We define the \emph{sample complexity} $SC_\A(\lambda,\theta)$  as the expected number of queries made by $\A$ to the blackbox given inputs $\lambda$ and $\theta$, where expectation is taken  the randomness inherent in the Monte-Carlo sampling from $\lambda$ as well as any other internal coins of $\A$.  We show that the worst-case sample complexity is not bounded by any function of $n$ and the approximation parameters  unless we allow bi-criteria loss in both optimality and incentive compatibility. More so, we show a stronger negative result for exactly incentive compatible algorithms:  the average sample complexity over $\theta \sim \lambda$ is also not bounded by a function of $n$ and the suboptimality parameter. Whereas our results imply that we should give up on exact incentive compatibility, we leave open the question of whether an optimal and $\eps$-incentive compatible algorithm exists with $\poly(n,\frac{1}{\eps})$  average case (but unbounded worst-case)  sample complexity.

\begin{theorem}\label{thm:hardness:blackbox}
	The following 
	hold for every algorithm $\A$ for  Bayesian Persuasion  in the black-box model:
	\begin{enumerate}[label=(\alph*)]
		\item If $\A$ is incentive compatible and $c$-optimal for $c< 1$, then for every integer $K$ there is a distribution $\lambda=\lambda(K)$ on 2 actions and 2 states of nature such that $\Ex_{\theta \sim \lambda} [SC_\A(\lambda,\theta) ] > K$. 
		\item If $\A$ is optimal and $c$-incentive compatible for $c<\frac{1}{4}$, then for every integer $K$ there is a distribution $\lambda=\lambda(K)$ on 3 actions and 3 states of nature, and $\theta$ in the support of $\lambda$,  such that  $SC_\A(\lambda,\theta)>K$.
	\end{enumerate}
\end{theorem}
Our proof of each part of this theorem involves constructing a pair of distributions $\lambda$ and $\lambda'$ which are arbitrarily close in statistical distance, but with the property that any algorithm with the postulated guarantees must distinguish between $\lambda$ and $\lambda'$. 
We defer the proof to Appendix~\ref{app:general:barriers}.

\section*{Acknowledgments}
We thank David Kempe for helpful comments on an earlier draft of this paper. We also thank the anonymous reviewers for helpful feedback and suggestions.


{
\bibliography{refer}
\bibliographystyle{abbrvnat}               
}

\appendix
\newpage

\section{Additional Discussion of Connections to Bayesian Mechanism Design}
\label{app:intro:related_bmd}


Section \ref{sec:iid}, which considers persuasion with independent and identically-distributed actions, relates to two ideas from auction theory. First, our symmetrization result in Section \ref{sec:iid:symmetry} is similar to  that of \citet{weinberg_symmetry}, but involves an additional ingredient which is necessary in our case: not only is the posterior type distribution for a recommended action (the winning bidder in the auction analogy) independent of the identity of the action, but so is the posterior type distribution of an unrecommended action (losing bidder). Second, our algorithm for computing the optimal scheme in Section \ref{sec:iid:optimal}  involves a connection  to Border's characterization of the space of feasible reduced-form single-item auctions \cite{Border91,Border07}, as well as its algorithmic properties \cite{Cai12a,Saeed12}. However, unlike in the case of single-item auctions, this connection hinges crucially on the symmetries of the optimal scheme, and fails to generalize to the case of persuasion with independent non-identical actions (analogous to independent non-identical bidders)  as we show in Section \ref{sec:indep}. We view this as evidence that persuasion and auction design --- while bearing similarities and technical connections --- are importantly different.

Section \ref{sec:indep} shows that our Border's theorem-based approach in Section \ref{sec:iid} can not be extended to independent non-identical actions.  Our starting point are the results of \citet{Gopalan15}, who rule out Border's-theorem like characterizations for a number of mechanism design settings by showing the \#P-hardness of computing the maximum expected  revenue or welfare. Our results similarly show that it is \#P hard to compute the maximum expected sender utility, but our reduction is much more involved. Specifically, whereas we also reduce from the \#P-hard problem of computing the Khintchine constant of a vector, unlike in \cite{Gopalan15} our reduction must go through the membership problem of a polytope which we use to encode the Khintchine constant computation. This detour seems unavoidable due to the different nature of the incentive-compatibility constraints placed on a signaling scheme.\footnote{In \cite{Gopalan15}, Myerson's characterization is used to show that optimal mechanism design in a public project setting directly encodes computation of the Khintchine constant. No analogous direct connection seems to hold here.} Specifically, we present an intricate reduction from membership testing in this ``Khintchine polytope'' to an optimal persuasion problem with independent actions.

Our algorithmic result for the black box model in Section \ref{sec:general} draws inspiration from, and is technically related to, the work in \cite{Cai12a, Saeed12,  Cai12b, Weinberg14} on algorithmically efficient mechanisms for multi-dimensional settings. Specifically, an alternative algorithm for our problem can be derived using the framework of \emph{reduced forms} and \emph{virtual welfare} of \citet{Cai12b} with significant additional work.\footnote{We thank an anonymous reviewer for pointing out this connection.} For this, a different reduced form is needed which allows for an unbounded ``type space'', and maintains the correlation information across actions necessary for evaluating the persuasion notion of incentive compatibility, which is importantly different from incentive compatibility in mechanism design. Such a reduced form exists, and the resulting algorithm is complex and invokes the ellipsoid algorithm as a subroutine. The algorithm we present here is much simpler and more efficient both in terms of runtime and samples from the distribution $\lambda$ over states of nature, with the main computational step being a single explicit linear program which solves for the optimal signaling scheme on a sample $\tilde{\lambda}$ from $\lambda$. The analysis of our algorithm is also more straightforward. This is possible in our setting due to our different notion of incentive compatibility, which permits reducing incentive compatibility on $\lambda$ to incentive compatibility on the sample $\tilde{\lambda}$ using the principle of deferred decisions.


\newpage

\section{Omissions from Section \ref{sec:iid}}
\label{app:iid}

\subsection{Symmetry of the Optimal Scheme (Theorem~\ref{thm:symmetry})}
\label{app:iid:symmetry}

To prove Theorem~\ref{thm:symmetry}, we need two closure properties of optimal signaling schemes --- with respect to permutations and convex combinations. We use $\pi$ to denote a permutation of $[n]$, and let $\SS_n$ denote the set of all such permutations. We define the permutation $\pi(\theta)$ of a state of nature $\theta \in [m]^n$ so that $(\pi(\theta))_j = \theta_{\pi(j)}$, and similarly the permutation of a signal $\sigma_i$ so that $\pi(\sigma_i) = \sigma_{\pi(i)}$. Given a signature  $\M=\{ (M^{\sigma_i},\sigma_i)\}_{i \in [n]}$, we define the permuted signature $\pi(\M) = \{ (\pi  M^{\sigma_i},\pi(\sigma_i))\}_{i \in [n]}$, where $\pi M$ denotes applying permutation $\pi$ to the rows of a matrix $M$.

\begin{lemma}
	\label{lem:Permu-opt}Assume the action payoffs  are i.i.d., and let $\pi \in \SS_n$ be an arbitrary permutation.  If $\M$ is the signature of a signaling scheme $\varphi$, then  $\pi(\M)$ is the signature of the scheme  $\varphi_\pi$ defined by $\varphi_\pi(\theta) = \pi(\varphi(\pi^{-1}(\theta)))$. Moreover, if $\varphi$ is incentive compatible and optimal, then so is $\varphi_\pi$.
\end{lemma}
\begin{proof}
	Let  $\M=\set{ (M^{\sigma},\sigma)}_{\sigma \in \Sigma}$ be the signature of $\varphi$, as given in the statement of the lemma. We first show that $\pi(\M) =\set{ (\pi M^{\sigma},\pi(\sigma))}_{\sigma \in \Sigma}$ is realizable as the signature of the scheme $\varphi_\pi$. By definition, it suffices to show that  $\sum_{\theta} \lambda(\theta) \varphi_\pi(\theta, \pi(\sigma)) M^\theta  =\pi M^\sigma$ for an arbitrary signal $\pi(\sigma)$.
	\begin{align*}
	\sum_{\theta} \lambda(\theta) \varphi_\pi(\theta, \pi(\sigma)) M^\theta  &= \sum_{\theta} \lambda(\theta) \varphi(\pi^{-1}(\theta), \sigma) M^\theta & \mbox{(by definition of $\varphi_\pi$)}  \\
	&= \pi \sum_{\theta \in \Theta} \lambda(\theta) \varphi(\pi^{-1}(\theta), \sigma) (\pi^{-1} M^\theta ) & \mbox{(by linearity of permutation)} \\
	&= \pi \sum_{\theta \in \Theta} \lambda(\theta) \varphi(\pi^{-1}(\theta), \sigma)  M^{\pi^{-1}(\theta)} & \mbox{} \\
	&= \pi \sum_{\theta \in \Theta} \lambda(\pi^{-1}(\theta)) \varphi(\pi^{-1}(\theta), \sigma)   M^{\pi^{-1}(\theta)} & \mbox{(Since $\lambda$ is i.i.d.)} \\
	&= \pi \sum_{\theta' \in \Theta} \lambda(\theta') \varphi(\theta', \sigma)   M^{\theta'} & \mbox{(by renaming $\pi^{-1}(\theta)$ to $\theta'$)} \\
	&=  \pi M^\sigma & \mbox{(by definition of $M^\sigma$)}
	\end{align*}
	
	Now, assuming $\varphi$ is incentive compatible, we check that $\varphi_\pi$ is incentive compatible by verifying the relevant inequality for its signature.
	\begin{align*}
	\bvec{\rho} \cdot (\pi M^{\sigma_i})_{\pi(i)} -  \bvec{\rho} \cdot (\pi M^{\sigma_i})_{\pi(j)} =    \bvec{\rho} \cdot  M^{\sigma_i}_i - \bvec{\rho} \cdot  M^{\sigma_i}_j  
	\geq 0
	\end{align*}
	Moreover, we show that the sender's utility is the same for $\varphi$ and $\varphi_\pi$, completing the proof.
	\begin{align*}
	\bvec{\xi} \cdot (\pi M^{\sigma_i})_{\pi(i)} &=  \bvec{\xi} \cdot ( M^{\sigma_i})_{i}
	\end{align*}
\end{proof}
\begin{lemma}
	\label{lem:Convex-opt}Let $t \in [0,1]$. If  $\A= (A^{\sigma_1},\ldots, A^{\sigma_n})$ is the signature of scheme $\varphi_A$,  and $\B= (B^{\sigma_1},\ldots, B^{\sigma_n})$ is the signature of a scheme $\varphi_B$, then  their convex combination $\C = (C^{\sigma_1},\ldots, C^{\sigma_n})$ with $C^{\sigma_i} = t A^{\sigma_i} + (1-t) B^{\sigma_i}$ is the signature of the scheme $\varphi_C$ which, on input $\theta$, outputs $\varphi_A(\theta)$ with probability $t$ and $\varphi_B(\theta)$ with probability $1-t$. Moreover, if $\varphi_A$ and $\varphi_B$ are both optimal and incentive compatible, then so is $\varphi_C$.
\end{lemma}
\begin{proof}
	This follows almost immediately from the fact that the optimization problem in Figure \ref{fig:optiid} is a linear program, with a convex feasible set and a convex family of optimal solutions. We omit the straightforward details.
\end{proof}

\subsubsection*{Proof of Theorem \ref{thm:symmetry}}
Given an optimal and incentive compatible signaling scheme $\varphi$ with signature  $\{ (M^{\sigma_i},\sigma_i)\}_{i \in [n]}$, we show  the existence of a symmetric optimal and incentive-compatible scheme  of the form
in Definition \ref{def:symmetry}. According to Lemma \ref{lem:Permu-opt}, for $\pi \in \SS_n$ the signature
$\{ (\pi M^{\sigma_i},\pi(\sigma_i))\}_{i \in [n]}$ ---  equivalently written as $\{ (\pi M^{\sigma_{\pi^{-1}(i)}}, \sigma_{i}\}_{i \in [n]}$ ---  corresponds to the optimal incentive compatible scheme $\varphi_\pi$.
Invoking Lemma \ref{lem:Convex-opt}, the signature
\[
\{ (A^{\sigma_i},\sigma_i)\}_{i \in [n]}= \{ (\frac{1}{n!}\sum_{\pi\in\SS_n} \pi M^{\sigma_{\pi^{-1}(i)}},\sigma_i)\}_{i \in [n]}
\]
also corresponds to an optimal and incentive compatible scheme, namely the scheme which draws a permutation $\pi$ uniformly at random, then signals according to $\varphi_\pi$.

Observe that the $i$th row of the matrix $\pi M^{\sigma_{\pi^{-1}(i)}}$
is the $\pi^{-1}(i)$th row of the matrix $M^{\sigma_{\pi^{-1}(i)}}$. Expressing $A^{\sigma_i}_i$ as a sum over permutations $\pi \in \SS_n$, and grouping the sum by $k=\pi^{-1}(i)$,  we can write
\begin{align*}
A^{\sigma_i}_i &= \frac{1}{n!}\sum_{\pi\in\SS_n} [\pi M^{\sigma_{\pi^{-1}(i)}}]_{i} \\
&= \frac{1}{n!}\sum_{\pi\in\SS_n}  M^{\sigma_{\pi^{-1}(i)}}_{\pi^{-1}(i)} \\
&= \frac{1}{n!} \sum_{k=1}^n    M^{\sigma_k}_{k} \cdot \left|\set{\pi \in \SS_n : \pi^{-1}(i) =k}\right| \\
&= \frac{1}{n!} \sum_{k=1}^n    M^{\sigma_k}_{k} \cdot (n-1)! \\
&= \frac{1}{n} \sum_{k=1}^n    M^{\sigma_k}_{k},
\end{align*}
%
%
which does not depend on $i$. Similarly, the $j$th row of the matrix $\pi M^{\sigma_{\pi^{-1}(i)}}$ is the $\pi^{-1}(j)$th row of the matrix $M^{\sigma_{\pi^{-1}(i)}}$. For $j \neq i$, expressing $A^{\sigma_i}_j$ as a sum over permutations $\pi \in \SS_n$, and grouping the sum by $k=\pi^{-1}(i)$ and $l = \pi^{-1}(j)$,  we can write
\begin{align*}
A^{\sigma_i}_j &= \frac{1}{n!}\sum_{\pi\in\SS_n} [\pi M^{\sigma_{\pi^{-1}(i)}}]_{j} \\
&= \frac{1}{n!}\sum_{\pi\in\SS_n}  M^{\sigma_{\pi^{-1}(i)}}_{\pi^{-1}(j)} \\
&= \frac{1}{n!} \sum_{k\neq l}    M^{\sigma_k}_{l} \cdot \left|\set{\pi \in \SS_n : \pi^{-1}(i) =k, \pi^{-1}(j)=l}\right| \\
&= \frac{1}{n!} \sum_{k \neq l}    M^{\sigma_k}_{l} \cdot (n-2)! \\
&= \frac{1}{n(n-1)} \sum_{k \neq l}    M^{\sigma_k}_{l},
\end{align*}
which does not depend on $i$ or $j$. Let 
\begin{eqnarray*}
	\bvec{x} & = & \frac{1}{n}\sum_{k=1}^{n} M^{\sigma_k}_{k}; \\
	\bvec{y} & = & \frac{1}{n(n-1)}\sum_{k \not = l} M^{\sigma_k}_{l}. 
\end{eqnarray*}
The signature $\{ (A^{\sigma_i},\sigma_i)\}_{i \in [n]}$ therefore describes  an optimal, incentive compatible, and symmetric scheme  with $s$-signature $(\bvec{x},\bvec{y})$.

\subsection{The Optimal Scheme}
\label{app:iid:opt}

\subsubsection*{Proof of Lemma \ref{lem:reduced_form}}
For the ``only if" direction, $||\bvec{x}||_1 =  \frac{1}{n}$ and $\bvec{x}+(n-1)\bvec{y}=\bvec{q}$ were established in Section \ref{sec:iid:symmetry}. To show that $\bvec{\tau}$ is a realizable symmetric reduced form for an allocation rule, let $\varphi$ be a signaling scheme with $s$-signature $(\bvec{x},\bvec{y})$. Recall from the definition of an $s$-signature that, for each $i \in [n]$, signal $\sigma_i$ has probability $1/n$, and  $n \bvec{x}$ is the posterior distribution of action $i$'s type conditioned on  signal $\sigma_i$. Now consider the following  allocation rule: Given a type profile $\theta \in [m]^n$ of the $n$ bidders,  allocate the item to bidder
$i$ with probability $\varphi(\theta,\sigma_{i})$ for any $i\in[n]$. By Bayes rule, 
\begin{align*}
\Pr [\text{$i$ gets  item}  | \text{$i$ has type $j$}] &=  \Pr [\text{$i$ has type $j$}  | \text{$i$ gets item}] \cdot \frac{\Pr [ \text{$i$ gets  item}]}{\Pr[\text{$i$ has type $j$}]}\\
&= nx_j \cdot \frac{1/n}{q_j} = \frac{x_j}{q_j}
\end{align*}
Therefore $\bvec{\tau}$ is indeed the reduced form of the described allocation rule.

For the ``if'' direction, let $\bvec{\tau}$, $\bvec{x}$, and $\bvec{y}$ be as  in the statement of the lemma, and consider an allocation rule $A$ with symmetric reduced form $\bvec{\tau}$. Observe that $A$ always allocates the item, since for each player $i \in [n]$ we have $\Pr [\text{$i$ gets the item}] = \sum_{j=1}^m q_j \tau_j = \sum_{j=1}^m x_j =\frac{1}{n}$. We define the direct signaling scheme $\varphi_A$  by $\varphi_A(\theta) = \sigma_{A(\theta)}$. Let $\M = (M^{\sigma_1}, \ldots, M^{\sigma_n})$ be the signature of $\varphi_A$. Recall that, for $\theta \sim \lambda$ and arbitrary  $i \in [n]$ and $j \in [m]$,  $M^{\sigma_i}_{ij}$ is the probability that  $\varphi_A(\theta)=\sigma_i$ and $\theta_i = j$; by definition, this equals the probability that $A$ allocates the item to player $i$ and her type is $j$, which is $\tau_j q_j = x_j$.   As a result, the signature $\M$ of $\varphi_A$ satisfies $M^{\sigma_i}_i = \bvec{x}$ for every action $i$. If $\varphi_A$ were symmetric, we would conclude that its $s$-signature is $(\bvec{x},\bvec{y})$ since every $s$-signature $(\bvec{x},\bvec{y}')$ must satisfy $\bvec{x}+(n-1)\bvec{y}'=\bvec{q}$ (see Section \ref{sec:iid:symmetry}). However, this is not guaranteed when the allocation rule $A$ exhibits some asymmetry. Nevertheless, $\varphi_A$ can be ``symmetrized'' into a signaling scheme $\varphi'_A$ which first draws a random permutation $\pi \in \SS_n$, and signals $\pi(\varphi_A(\pi^{-1}(\theta)))$. That $\varphi'_A$ has $s$-signature $(\bvec{x},\bvec{y})$ follows a similar argument to that used in the proof of Theorem \ref{thm:symmetry}, and we therefore omit the details here.

Finally, observe that the description of $\varphi'_A$ above is constructive assuming black-box access to~$A$, with runtime overhead that is polynomial in $n$ and $m$.

\subsubsection*{Proof of Lemma \ref{lem:iid_solve}}

By Lemma \ref{lem:reduced_form}, we can re-write LP~\eqref{lp:iid:opt} as follows:
\begin{lp}
	
	\maxi{n\bvec{\xi} \cdot \bvec{x} }   
	\st
	\con{ \bvec{\rho} \cdot \bvec{x} \geq \bvec{\rho} \cdot \bvec{y} }
	\con{ \bvec{x} + (n-1) \bvec{y} = \bvec{q} }
	\con{ ||\bvec{x}||_{1}=\frac{1}{n} }
	\con{ (\frac{x_1}{q_1},....,\frac{x_m}{q_m}) \text{ is a realizable symmetric reduced form} }
\end{lp}

From \cite{Border91,Border07,Cai12a,Saeed12},  we know that the family of all the realizable symmetric reduced forms constitutes a polytope, and moreover that this polytope admits an efficient  separation oracle. The runtime of this oracle is polynomial in $m$ and $n$, and as a result the above linear program can be solved in $poly(n,m)$ time using the Ellipsoid method. 

\subsection{A Simple $(1-1/e)$-approximate Scheme}
\label{app:iid:apx}
\subsubsection*{Proof of Theorem \ref{thm:iid_apx}}

Given a binary signal $\sigma=(o_1,\ldots,o_n) \in\{\high,\low\}^{n}$, the posterior type distribution for an action equals $n\bvec{x}^{*}$ if the corresponding component signal is $\high$, and equals $n\bvec{y}^{*}$ if the component signal is $\low$. This is simply a consequence of the independence of the action types,  the fact that the different component signals are chosen independently, and Bayes' rule.  The constraint $\bvec{\rho} \cdot  \bvec{x}^* \geq \bvec{\rho} \cdot  \bvec{y}^*$ implies that the receiver prefers actions $i$ for which $o_i=\high$, any one of which induces an expected utility of $n\bvec{\rho} \cdot \bvec{x}^*$ for the receiver and  $n\bvec{\xi} \cdot \bvec{x}^*$ for the sender. The latter quantity matches  the optimal value of LP~\eqref{lp:iid:apx}. The constraint $||\bvec{x}||_1= \frac{1}{n}$ implies that each component signal is $\high$ with probability $\frac{1}{n}$, independently. Therefore, the probability that at least one component signal is $\high$ equals $1 - (1-\frac{1}{n})^n \geq 1-\frac{1}{e}$. Since payoffs are nonnegative, and since a rational receiver selects a $\high$ action when one is available, the sender's overall expected utility is at least a $1-\frac{1}{e}$ fraction of the optimal value of LP \eqref{lp:iid:apx}. 



\newpage

\section{Proof of Theorem \ref{thm:hardness}}
\label{app:indep}

This section is devoted to proving Theorem \ref{thm:hardness}. Our proof starts from the ideas of  \citet{Gopalan15}, who show the \#P-hardness for revenue or welfare maximization in several mechanism design problems. In one case, \cite{Gopalan15} reduce from the $\#P$-hard problem of computing the \emph{Khintchine constant} of a vector. Our reduction also starts from this problem, but is much more involved: First, we exhibit a polytope which we term \emph{Khintchine polytope}, and show that computing the Khintchine constant reduces to linear optimization over the Khintchine polytope. Second, we present a reduction from the membership problem for the Khintchine polytope to the computation of optimal sender utility in a particularly-crafted instance of persuasion with independent actions. Invoking the polynomial-time equivalence between membership checking and optimization (see, e.g., \cite{Groetschel88}), we conclude the \#P-hardness of our problem. The main technical challenge we overcome is in the second step of our proof: given a point $x$ which may or may not be in the Khintchine polytope $\K$, we construct a persuasion instance and a threshold $T$ so that points in $\K$ encode signaling schemes, and the optimal sender utility is at least $T$ if and only if $x \in \K$ and the scheme corresponding to $x$ results in sender utility $T$.


\subsection*{The Khintchine Polytope}
We start by defining  the \emph{Khintchine problem}, which is shown to be \#P-hard in \cite{Gopalan15}.

\begin{definition}
	(Khintchine Problem) Given a vector $a \in \RR^n$, compute the \emph{Khintchine constant} $K(a)$ of $a$, defined as follows:
	\begin{equation*}
	K(a) = \Ex_{\theta \sim \{\pm 1\}^n} [ |\theta \cdot a|],
	\end{equation*} 
	where $\theta$ is drawn \emph{uniformly} at random from $\{\pm 1\}^n$.		
\end{definition}

To relate the Khintchine problem to Bayesian persuasion, we begin with a persuasion instance with $n$ \emph{i.i.d.}~actions. Moreover, there are only two action types,\footnote{Recall from Section \ref{sec:iid} that each type is associated with a pair $(\xi,\rho)$, where $\xi$ [$\rho$] is the payoff to the sender [receiver] if the receiver takes an action of that type. } which we refer to as \emph{type -1} and \emph{type +1}. The state of nature is a uniform random draw from the set $\{\pm 1 \}^n$, with the $i$th entry specifying the type of action~$i$. It is easy to see that these actions are \emph{i.i.d.}, with marginal probability  $\frac{1}{2}$ for each type. We call this instance the \emph{Khintchine-like} persuasion setting.  As in Section \ref{sec:iid}, we still use the \emph{signature} to capture the payoff-relevant features of a signaling scheme. A signature for the Khintchine-like persuasion problem is of the form $\M = (M^1,...,M^n)$ where $M^i \in \RR^{n \times 2}$ for any $i \in [n]$.   We pay special attention to signaling schemes which use only \emph{two} signals, in which case we represent them using a \emph{two-signal signature} of the form $(M^1,M^2) \in \RR^{n \times 2} \cross \RR^{n \times 2}$. Recall that such a signature is \emph{realizable} if there is a signaling scheme which uses only two signals, with the property that  $M^i_{jt}$ is the joint probability of the $i$th signal and the event that action $j$ has type $t$. We now define the \emph{Khintchine polytope}, consisting of a convex family of two-signal signatures.
\begin{definition}
	The \emph{Khintchine polytope}  is the family $\mathcal{K}(n)$ of  \emph{realizable}  two-signal signatures $(M^1,M^2)$ for the Khintchine-like persuasion setting which satisfy the additional constraints $M^1_{i,1} + M^1_{i,2} = \frac{1}{2} \,  \forall i\in[n]$.
\end{definition}

We sometimes use $\mathcal{K}$ to denote the Khintchine polytope $\mathcal{K}(n)$ when the dimension $n$ is clear from the context. Note that the constraints $M^1_{i,1} + M^1_{i,2} = \frac{1}{2}, \, \forall i\in[n]$ state that the first signal should be sent with probability $\frac{1}{2}$ (hence also the second signal).   We now show that optimizing over the Khintchine polytope is $\#P$-hard by reducing the Kintchine problem to Linear program \eqref{lp:Ka}.

\begin{figure}
	\centering
	\begin{lp}
		\label{lp:Ka}
		\maxi{\sum_{i=1}^n a_i (M^+_{i,+1} - M^+_{i,-1}) -  \sum_{i=1}^n a_i (M^-_{i,+1} - M^-_{i,-1})}
		\st
		\con{(M^+,M^-) \in \K(n)}
	\end{lp}
	Linear program for computing the Khintchine constant $K(a)$ for $a \in \RR^n$ 
\end{figure}

\begin{lemma}\label{lem:PolytopeHard}
	General linear optimization over the Khintchine polytope $\mathcal{K}$ is $\#P$-hard.
\end{lemma}
\begin{proof}
	For any given $a\in \RR^n$, we reduce the computation of $K(a)$ -- the Khintchine constant for $a$ -- to a linear optimization problem over the Khintchine polytope $\mathcal{K}$. Since our reduction will use two signals $\sigma_+$ and $\sigma_{-}$ which correspond  to the sign of $\theta \cdot a$,  we will  use $(M^+,M^-)$ to denote the two matrices in the signature in lieu of $(M^1,\, M^2)$. Moreover, we use the two action types $+1$ and $-1$ to index the columns of each matrix. For example, $M^+_{i,-1}$ is the joint probability of signal $\sigma_+$ and the event that the $i$th action has type $-1$.
	
	We claim that the Kintchine constant $K(a)$ equals the  optimal objective value of the  implicitly-described linear program~\eqref{lp:Ka}. We denote this optimal objective value  by $OPT(LP~\eqref{lp:Ka})$.			
	%
	%
	We first prove that $K(a) \leq OPT(LP~\eqref{lp:Ka})$. Consider a signaling scheme $\varphi$ in the Kintchine-like persuasion setting which simply outputs $\sigma_{sign(\theta \cdot a)}$ for each state of nature $\theta \in \set{\pm 1}^n$ (breaking tie uniformly at random if $\theta \cdot a = 0$). Since $\theta$ is drawn uniformly from $\set{\pm 1}^n$ and $sign(\theta \cdot a) = -sign( -\theta \cdot a)$, this scheme outputs each of the signals $\sigma_{-}$ and $\sigma_+$ with probability $\frac{1}{2}$. 
	Consequently, the two-signal signature of $\varphi$ is a point in $\mathcal{K}$. Moreover, evaluating the objective function of  LP ~\eqref{lp:Ka} on the two-signal signature $(M^+,M^-)$ of $\varphi$ yields  $K(a)=\Ex_{\theta} [|\theta \cdot a|]$, as shown below.
	\begin{align*}
	\Ex_{\theta} [|\theta \cdot a|] &= \Ex_{\theta} [\theta \cdot a | \sigma_+] \cdot \Pr( \sigma_+) + \Ex_{\theta} [-\theta \cdot a | \sigma_{-}] \cdot \Pr( \sigma_{-}) \\
	&= \sum_{i=1}^{n} a_i \Ex_{\theta} [\theta_i | \sigma_+] \cdot \Pr( \sigma_+) - \sum_{i=1}^{n} a_i \Ex_{\theta} [\theta_i | \sigma_{-}] \times \Pr( \sigma_{-}) \\
	&=  \sum_{i=1}^{n} \bigg(  a_i[\Pr(\theta_i = 1 | \sigma_+)  - \Pr (\theta_i = -1 | \sigma_+)] \cdot \Pr( \sigma_+)   \bigg)\\
	& \qquad - \sum_{i=1}^{n} \bigg( a_i [\Pr(\theta_i = 1 | \sigma_{-}) 
	- \Pr (\theta_i = -1 | \sigma_{-})] \cdot \Pr( \sigma_{-}) \bigg) \\
	&= \sum_{i=1}^{n} \bigg(a_i [\Pr(\theta_i = 1, \sigma_+) - \Pr (\theta_i = -1, \sigma_+)] \bigg)
	- \sum_{i=1}^{n} \bigg( a_i [\Pr(\theta_i = 1 , \sigma_{-}) 
	- \Pr (\theta_i = -1, \sigma_{-})] \bigg) \\
	&= \sum_{i=1}^{n} a_i [M^+_{i,+1} 
	- M^+_{i,-1}] - \sum_{i=1}^{n} a_i [M^-_{i,+1} 
	- M^-_{i,-1}]
	\end{align*}
	This concludes the proof that $K(a) \leq OPT(LP~\eqref{lp:Ka})$.
	
	Now we prove $K(a) \geq OPT(LP~\eqref{lp:Ka})$. Take \emph{any} signaling scheme which uses only two signals $\sigma_+$ and $\sigma_-$, and let $(M^+, M^-)$ be its two-signal signature. Notice, however, that $\sigma_+$ now is only the ``name" of the signal, and does not imply that $\theta \cdot a$ is positive.  Nevertheless, it is still valid to reverse the above derivation until we reach
	\begin{eqnarray*}
		\sum_{i=1}^{n} a_i [M^+_{i,+1} 
		- M^+_{i,-1}] - \sum_{i=1}^{n} a_i [M^-_{i,+1} 
		- M^-_{i,-1}] &=& \Ex_{\theta} [\theta \cdot a | \sigma_+] \cdot \Pr( \sigma_+) + \Ex_{\theta} [-\theta \cdot a | \sigma_{-}] \cdot \Pr( \sigma_{-}). 
	\end{eqnarray*}   
	Since $\theta \cdot a $ and $-\theta \cdot a $ are each no greater than $|\theta \cdot a|$, we  have
	\begin{align*}
	\Ex_{\theta} [\theta \cdot a | \sigma_+] \cdot \Pr( \sigma_+) + \Ex_{\theta} [-\theta \cdot a | \sigma_{-}] \cdot \Pr( \sigma_{-}) &\leq \Ex_{\theta} [|\theta \cdot a| \mid \sigma_+] \cdot \Pr( \sigma_+) + \Ex_{\theta} [|\theta \cdot a |\mid  \sigma_{-}] \cdot \Pr( \sigma_{-})\\
	&=  \Ex_{\theta} [|\theta \cdot a|] = K(a).
	\end{align*}
	That is, the objective value of LP~\eqref{lp:Ka} is upper bounded by $K(a)$, as needed.
\end{proof}

Before we proceed to present the reduction from the membership problem for $\mathcal{K}$ to optimal persuasion, we point out an interesting corollary of Lemma \ref{lem:PolytopeHard}.
\begin{corollary}\label{cor:IIDhard}
	Let $\mathcal{P}$ be the polytope of realizable signatures for a persuasion problem with $n$ \emph{i.i.d.} actions and $m$  types (see Section \ref{sec:iid}). Linear optimization over $\mathcal{P}$ is $\#P$-hard, and this holds even when $m=2$. 
\end{corollary}
\begin{proof}
	Consider the Khintchine-like persuasion setting. It is easy to see that the Khintchine polytope $\mathcal{K}$ can be obtained from $\P$ by adding the constraints $M^{\sigma_i} = 0$ for  $i \geq 3$ and $M^{\sigma_1}_{i,1}+M^{\sigma_1}_{i,2} = \frac{1}{2}$ for $i \in [n]$, followed by a simple projection.  Therefore, the membership problem for $\mathcal{K}$ can be reduced in polynomial time to the membership problem for $\mathcal{P}$, since the additional linear constraints can be explicitly checked in polynomial time.  By the polynomial-time equivalence between optimization and membership, it follows that general linear optimization over $\mathcal{P}$ is $\#P$-hard. 
\end{proof}

\begin{remark}
	It is interesting to compare Corollary \ref{cor:IIDhard} to single item auctions with i.i.d. bidders, where the problem does admit a polynomial-time separation oracle for the polytope of realizable signatures via Border's Theorem \cite{Border91,Border07} and its algorithmic properties \cite{Cai12a,Saeed12}. In contrast,  the polytope of realizable signatures for Bayesian persuasion is $\#$P-hard to optimize over. Nevertheless, in Section \ref{sec:iid} we were indeed able to compute the optimal signaling scheme and sender utility for persuasion with i.i.d. actions. Corollary~\ref{cor:IIDhard} conveys that it was crucial for our algorithm to exploit the special structure of the persuasion objective and the symmetry of the optimal scheme, since optimizing a general objective over $\mathcal{P}$ is \#P-hard.         		
\end{remark}

\subsection*{Reduction}
We now present a reduction from the membership problem for the Khintchine polytope to the computation of optimal sender utility for persuasion with independent actions. As the output of our reduction, we construct a persuasion instance of the following form. There are $n+1$ actions. Action $0$ is \emph{special} -- it deterministically results in sender utility $\epsilon$ and receiver utility $0$. Here, we think of $\epsilon > 0$ as being small enough for our arguments to go through. The other $n$ actions are \emph{regular}. Action $i>0$ \emph{independently} results in sender utility $- a_i$ and receiver utility $a_i$ with probability $\frac{1}{2}$ (call this the type $1_i$), or sender utility $-  b_i$ and receiver utility $b_i$ with probability $\frac{1}{2}$ (call this the type $2_i$).  Note that the sender and receiver utilities are \emph{zero-sum} for both types.  Notice that, though each regular action's type distribution is uniform over its two types, the actions here are \emph{not} identical because the associated payoffs --- specified by $a_i$ and $b_i$ for each action $i$ ---  are different for different actions. Since the special action is deterministic and the probability of its (only) type is $1$ in any signal, we  can interpret any $(M^1,M^2)\in \mathcal{K}(n)$ as a two-signal signature for our persuasion instance (the row corresponding to the special action $0$ is implied). For example, $M^1_{i,2}$ is the joint probability of the first signal and the event that  action $i$ has type $2_i$. Our goal is to reduce membership checking for $\mathcal{K}(n)$ to computing the optimal expected sender utility for a persuasion  instance with carefully chosen parameters $\set{a_i}_{i=1}^n$, $\set{b_i}_{i=1}^n$, and~$\epsilon$.

In relating optimal persuasion to the Khintchine polytope,
there are two main difficulties: (1) $\mathcal{K}$ consists of two-signal signatures, so there should be an optimal scheme to our persuasion instance which uses only two signals; (2) To be consistent with the definition of $\mathcal{K}$, such an optimal scheme should send each signal with probability exactly $\frac{1}{2}$. We will design specific $\epsilon,a_i, b_i$ to accomplish both goals.


For notational convenience, we will again use $(M^+,M^-)$ to denote a typical element in $\mathcal{K}$ instead of $(M^1,M^2)$ because, as we will see later, the two constructed signals will induce positive and negative sender utilities, respectively.  Notice that there are only $n$ degrees of freedom in $(M^+,M^-)\in \mathcal{K}$. This is because $M^+ + M^-$ is the all-$\frac{1}{2}$ matrix in $\RR^{n\times2}$, corresponding to the prior distribution of states of nature (by the definition of realizable signatures). Moreover, $M^+_{i,1}+M^-_{i,2} = \frac{1}{2}$ for all $i \in [n]$ (by the definition of $\mathcal{K}$).  Therefore, we must have 
\begin{equation*}
M^{+}_{i,1}=M^-_{i,2} = \frac{1}{2} - M^{+}_{i,2} = \frac{1}{2} - M^{-}_{i,1}.
\end{equation*}
This implies that we can parametrize signatures $(M^+,M^-) \in \K$ by a vector $\bvec{x} \in [0,\frac{1}{2}]^n$, where $M^+_{i,1} = M^{-}_{i,2}=x_i$ and 
$M^+_{i,2} = M^-_{i,1} = \frac{1}{2} - x_i$ for each $i \in [n]$. For any $\bvec{x} \in [0,\frac{1}{2}]^{n}$, let $\M(\bvec{x})$ denote the signature $(M^+,M^-)$ defined by $\bvec{x}$ as just described. 

We can now restate the membership problem for $\K$ as follows: given $\bvec{x} \in [0,\frac{1}{2}]^n$, determine whether $\M(\bvec{x}) \in \K$. When any of the entries of $\bvec{x}$ equals $0$ or $\frac{1}{2}$ this problem is trivial,\footnote{If $x_i$ is $0$ or $\frac{1}{2}$, then $\M(x) \in \mathcal{K}$ if and only if $x_j =\frac{1}{4}$ for all $j \neq i$. This is because the corresponding signaling scheme must choose its signal based solely on the type of action $i$.} so we assume without loss of generality that $\bvec{x} \in (0,\frac{1}{2})^n$. Moreover, when $x_i=\frac{1}{4}$ for some $i$, it is easy to see that a signaling scheme with signature $\M(\bvec{x})$, if one exists, must choose its signal independently of the type of action $i$, and therefore $\M(\bvec{x}) \in \K(n)$ if and only if $\M(\bvec{x_{-i}}) \in \K(n-1)$. This allows us to assume without loss of generality that $x_i \neq \frac{1}{4}$ for all $i$.

Given $\bvec{x} \in (0,\frac{1}{2})^n$ with $x_i \neq \frac{1}{4}$ for all $i$, we construct specific $\epsilon$ and $a_i,b_i$ for all $i$ such that we can determine whether $\M(\bvec{x})\in \mathcal{K}$ by simply looking at the optimal sender utility in the corresponding persuasion instance. We choose parameters $a_i$ and $b_i$ to satisfy the following two equations.
\begin{eqnarray}\label{eq:ConstructCond1}
&& x_i a_i + (\frac{1}{2} - x_i) b_i = 0. \\ \label{eq:ConstructCond2}
&& (\frac{1}{2} - x_i)a_i + x_i b_i = \frac{1}{2}.
\end{eqnarray}
We note that the above linear system always has a solution when $x_i \neq \frac{1}{4}$, which we assumed previously. We make two observations about our choice of $a_i$ and $b_i$. First, the \emph{prior} expected receiver utility $\frac{1}{2}(a_i+b_i)$  equals $\frac{1}{2}$ for all actions $i$ (by simply  adding Equation \eqref{eq:ConstructCond1} and \eqref{eq:ConstructCond2}). Second, $a_i$ and $b_i$ are both non-zero, and this follows easily  from our assumption that $x_i \in (0, \frac{1}{2})$. 

Now we show how to determine whether $\M(\bvec{x})\in \mathcal{K}$ by only examining the optimal sender utility in the constructed persuasion instance.  We start by showing that restricting to two-signal schemes is without loss of generality in our instance.

\begin{lemma}\label{lem:TwoSignal}
	There exists an optimal incentive-compatible signaling scheme which uses at most \emph{two} signals: one signal recommends the special action, and the other recommends some regular action.
\end{lemma} 
\begin{proof}
	Recall that an optimal incentive-compatible scheme uses $n+1$ signals, with signal $\sigma_i$ recommending action $i$ for $i=0,1,...,n$. Fix such a scheme, and let $\alpha_i$ denote the probability of signal $\sigma_i$. Signal $\sigma_i$ induces posterior expected receiver utility $r_j(\sigma_i)$ and sender utility $s_j(\sigma_i)$ for each action $j$. For a regular action  $j \neq 0$, we have $s_j(\sigma_i) = - r_j(\sigma_i)$ for all $i$ due to the zero-sum nature of our construction.  Notice that $r_i(\sigma_i) \geq 0$ for all regular actions $i \neq 0$, since otherwise the receiver would prefer action $0$ over action $i$. Consequently, for each signal $\sigma_i$ with $i \neq 0$, the receiver derives non-negative utility and the sender derives non-positive utility.
	
	We claim that merging signals $\sigma_1, \sigma_2, \ldots, \sigma_n$ --- i.e., modifying the signaling scheme to output the same signal $\sigma^*$ in lieu of each of them --- would not decrease the sender's expected utility. Recall that incentive compatibility implies that $r_i(\sigma_i)=\max_{j=0}^n r_j(\sigma_i)$.  Using Jensen's inequality, we get  
	\begin{equation}\label{eq:MergeBetter}
	\sum_{i=1}^n \alpha_i r_i(\sigma_i) \geq  \max_{j=0}^n \left[ \sum_{i=1}^{n} \alpha_i  r_j(\sigma_i)  \right].
	\end{equation} 
	If the maximum in the right hand side expression of \eqref{eq:MergeBetter} is attained at $j^*=0$, the receiver will choose the special action $0$ when presented with the merged signal $\sigma^*$. Recalling that $s_i(\sigma_i)$ is non-positive for $i \neq 0$, this can only improve the sender's expected utility. Otherwise, the receiver chooses a regular action $j^* \neq 0$ when presented with $\sigma^*$, resulting in a total contribution of $ \sum_{i=1}^{n} \alpha_i  r_{j^*}(\sigma_i)$ to the receiver's expected utility from the merged signal, down from the total contribution of $\sum_{i=1}^n \alpha_i r_i(\sigma_i)$ by the original signals $\sigma_1,\ldots,\sigma_n$. Recalling the zero-sum nature of our construction for regular actions, the merged signal $\sigma^*$ contributes $\sum_{i=1}^{n} \alpha_i  s_{j^*}(\sigma_i) = - \sum_{i=1}^{n} \alpha_i  r_{j^*}(\sigma_i)$ to the sender's expected utility, up from a total contribution of $\sum_{i=1}^n \alpha_i s_i(\sigma_i) = - \sum_{i=1}^n \alpha_i r_i(\sigma_i)$ by the original signals $\sigma_1,\ldots,\sigma_n$. Therefore, the sender is not worse off by merging the signals. Moreover, interpreting $\sigma^*$ as a recommendation for action $j^*$ yields incentive compatibility.
	%
	%
\end{proof}

Therefore, in characterizing the optimal solution to our constructed persuasion instance, it suffices to analyze two-signal schemes of the the form guaranteed by Lemma~\ref{lem:TwoSignal}. For such a scheme, we denote the  signal that recommends the special action $0$ by $\sigma_+$ (indicating that the sender derives positive utility $\epsilon$), and denote the other signal by $\sigma_-$ (indicating that the sender derives negative utility, as we will show). 	For convenience, in the following discussion we use the expression ``payoff from a signal" to signify the expected payoff of a player conditioned on that signal  multiplied by the probability of that signal. For example, the \emph{sender's expected payoff from signal $\sigma_-$} equals the sender's expected payoff conditioned on signal $\sigma_-$ multiplied by the overall probability that the scheme outputs $\sigma_-$, assuming the receiver follows the scheme's (incentive compatible) recommendations. We also use the expression ``payoff from an action in a signal'' to signify the posterior expected payoff of a player for that action conditioned on the signal, multiplied by the probability that the scheme outputs the signal. For example, the \emph{receiver's expected payoff from action $i$ in signal $\sigma_+$} equals  $ \alpha_+ \cdot r_i(\sigma_+)$, where $r_i(\sigma_+)$ is the receiver's posterior expected payoff from action $i$ given signal $\sigma_+$, and $\alpha_+$ is the overall probability of signal $\sigma_+$.

\begin{lemma}\label{lem:SenderUpperBound}
	Fix an incentive-compatible scheme with signals $\sigma_-$ and $\sigma_+$ as described above. The sender's expected payoff from signal $\sigma_-$ is at most $-\frac{1}{2} $. Moreover, if the sender' expected payoff from $\sigma_-$ is exactly $-\frac{1}{2}$, then for each regular action $i$ the expected payoff of both the sender and the receiver from action $i$ in signal $\sigma_+$ equals $0$.
\end{lemma}
\begin{proof}
	Assume that signal $\sigma_+$ [$\sigma_-$] is sent with probability $\alpha_+$ [$\alpha_-$] and induces posterior expected receiver payoff $r_i(\sigma_+) $ [$r_i(\sigma_-)$] for each action $i$. Recall from our construction that the \emph{prior} expected payoff of each regular action $i \neq 0$ equals  $\frac{1}{2}a_i + \frac{1}{2}b_i=\frac{1}{2}$. Since the prior expectation must equal the expected posterior expectation, it follows that  $\alpha_+ \cdot r_i(\sigma_+) + \alpha_- \cdot r_i(\sigma_-) = \frac{1}{2}$ when $i$ is regular. The receiver's reward from the special action is deterministically $0$, and therefore incentive compatibility implies that $r_i(\sigma_+) \leq 0$ for each regular action $i$. It follows that $\alpha_- \cdot r_i(\sigma_-) = \frac{1}{2} - \alpha_+ \cdot r_i(\sigma_+) \geq \frac{1}{2}$ for regular actions $i$. In other words, the receiver's expected payoff from each regular action  in signal $\sigma_-$ is at least $\frac{1}{2}$. By the zero-sum nature of our construction, the sender's expected payoff from each regular action in signal $\sigma_-$ is at most $-\frac{1}{2}$. Since $\sigma_-$ recommends a regular action, we conclude that the sender's expected payoff from $\sigma_-$ is at most  $-\frac{1}{2}$.
	
	Now assume that the sender's expected payoff from $\sigma_-$ is exactly $-\frac{1}{2}$. By the zero-sum property, incentive compatibility, and the above-established fact that $\alpha_- \cdot r_i(\sigma_-) \geq \frac{1}{2}$ for regular actions $i$, it follows that the receiver's expected payoff from each regular action in signal $\sigma_-$ is \emph{exactly} $\frac{1}{2}$.  Recalling that $\alpha_+ \cdot r_i(\sigma_+) + \alpha_- \cdot r_i(\sigma_-) = \frac{1}{2}$ when $i$ is regular, we conclude that the receiver's expected payoff from a regular action in signal $\sigma_+$ equals $0$. By the zero-sum property for regular actions, the same is true for the sender.

	
\end{proof}

The key to the remainder of our reduction is to choose a small enough value for the parameter $\epsilon$ --- the sender's utility from the special action --- so that the optimal signaling scheme satisfies the property mentioned in Lemma \ref{lem:SenderUpperBound}: The sender's expected payoff from signal $\sigma_-$ is exactly equal to its maximum possible value of $-\frac{1}{2}$. In other words, we must make $\epsilon$  so small so that the sender prefers to not sacrifice \emph{any} of her payoff from $\sigma_-$ in order to gain utility from the special action recommended by $\sigma_+$.  Notice that this upper bound of $-\frac{1}{2}$ is indeed achievable: the uninformative signaling scheme which recommends an arbitrary regular action has this property. 	We now show that  a ``small enough" $\epsilon$ indeed exists. The key idea behind this existence proof is the following: We start with a signaling scheme which maximizes the sender's payoff from $\sigma_-$ at $-\frac{1}{2}$, and moreover corresponds to a vertex of the polytope of incentive-compatible signatures. When $\eps>0$ is smaller than the ``bit complexity'' of the vertices of this polytope, moving to a different vertex  --- one with lower sender payoff from $\sigma_-$ ---  will result in more utility loss from $\sigma_-$ than utility gain from $\sigma_+$. We show that $\eps >0$ with polynomially many bits suffices, and can be computed in polynomial time.


Let $\mathcal{P}_2$ be the family of all \emph{realizable} two-signal signatures (again, ignoring action $0$). It is easy to see that $\mathcal{P}_2$ is a polytope, and importantly, all entries of any vertex of $\mathcal{P}_2$ are integer multiples of $\frac{1}{2^n}$. This is because every vertex of $\mathcal{P}_2$ corresponds to a deterministic signaling scheme which partitions the set of states of nature, and every state of nature occurs with probability $1/2^n$. As a result, all vertices of $\mathcal{P}_2$ have $\mathcal{O}(n)$ bit complexity. 

To ease our discussion, we use a compact representation for points in $\mathcal{P}_2$. In particular, any point in $\mathcal{P}_2$ can be captured by $n+1$ variables: variable $p$ denotes the probability of sending signal $\sigma_+$, and variable $y_i$  denotes the joint probability of signal $\sigma_+$ and the event that action $i$ has type $1_i$. It follows that joint probability of type $2_i$ and signal $\sigma_+$ is $p-y_i$, and the probabilities associated with signal  $\sigma_-$ are determined by the constraint that $M^++M^-$ is the all-$\frac{1}{2}$ matrix. With some abuse of notation, we use $\M(p,\bvec{y}) = (M^+,M^-)$ to denote the  signature in $\mathcal{P}_2$ corresponding to the probability $p$ and n-dimensional vector $\bvec{y}$. Now we consider the following two linear programs.

\begin{lp}\label{lp:Usigma+-}
	\maxi{p\epsilon + u}
	\st
	\con{\M(p,\bvec{y}) \in \mathcal{P}_2}
	\qcon{y_i a_i + (p-y_i)b_i \leq 0}{i = 1,\ldots,n}
	\qcon{u \leq -[(\frac{1}{2}-y_i)a_i + (\frac{1}{2}-p+y_i)b_i]}{ i = 1,\ldots,n}
\end{lp}
\begin{lp}\label{lp:Usigma+}
	\maxi{u}
	\st
	\con{\M(p,\bvec{y}) \in \mathcal{P}_2}
	\qcon{y_i a_i + (p-y_i)b_i \leq 0}{i = 1,\ldots,n}
	\qcon{u \leq -[(\frac{1}{2}-y_i)a_i + (\frac{1}{2}-p+y_i)b_i]}{ i = 1,\ldots,n}
\end{lp}
Linear programs \eqref{lp:Usigma+-} and \eqref{lp:Usigma+} are identical except for the fact that the objective of LP \eqref{lp:Usigma+-} includes the additional term $p\epsilon$. LP \eqref{lp:Usigma+-} computes precisely the optimal expected sender utility in our constructed persuasion instance: The first set of inequality constraints are the incentive-compatibility constraints for the signal $\sigma_+$ recommending action $0$; The second set of inequality constraints state that the sender's payoff from signal $\sigma_-$ is the minimum  among all actions, as implied by the zero-sum nature of our construction; The objective is the sum of the sender's payoffs from signals $\sigma_+$ and $\sigma_-$. Notice that the incentive-compatibility constraints for signal $\sigma_-$, namely $(\frac{1}{2}-y_i)a_i + (\frac{1}{2}-p+y_i)b_i \geq 0$ for all $i \neq 0$, are implicitly satisfied  because $\frac{1}{2}a_i + \frac{1}{2}b_i = \frac{1}{2}$  by our construction and $ (\frac{1}{2}-y_i)a_i + (\frac{1}{2}-p+y_i)b_i = \frac{1}{2}a_i + \frac{1}{2}b_i - [y_i a_i + (p-y_i)b_i] \geq \frac{1}{2} - 0 > 0$. On the other hand, LP \eqref{lp:Usigma+} maximizes the sender's expected payoff from signal $\mathit{\sigma_-}$. Observe that the optimal objective value of LP \eqref{lp:Usigma+} is precisely $-\frac{1}{2}$ because $u \leq -[(\frac{1}{2}-y_i)a_i + (\frac{1}{2}-p+y_i)b_i] \leq -\frac{1}{2}$ for all $i\neq 0$,  and equality is attained, for example, at $p=0$ and $\bvec{y}=0$. 

Let $\tilde{\mathcal{P}_2}$ be the set of all feasible $(u,\M(p,\bvec{y}))$ for LP \eqref{lp:Usigma+-} (and LP \eqref{lp:Usigma+}). Obviously, $\tilde{\mathcal{P}_2}$ is a polytope. We now argue that all vertices of $\tilde{\mathcal{P}_2}$ have  bit complexity polynomial in $n$ and the bit complexity of  $\bvec{x} \in (0,\frac{1}{2})^n$. In particular, denote  the bit complexity of $\bvec{x}$ by $\ell$. Since  $a_i,b_i$ are computed by a two-variable two-equation linear system involving $x_i$ (Equations \eqref{eq:ConstructCond1} and \eqref{eq:ConstructCond2}), they each have $O(\ell)$ bit complexity. Consequently, all the explicitly described facets of $\tilde{\mathcal{P}_2}$ have $\O(\ell)$ bit complexity. Moreover, since each vertex of $\mathcal{P}_2$ has $\mathcal{O}(n)$ bit complexity, each facet of $\mathcal{P}_2$ then has $\mathcal{O}(n^3)$ bit complexity, i.e., the coefficients of inequalities that determine the facets have $\mathcal{O}(n^3)$ bit complexity. This is due to the fact that facet complexity of a rational polytope is upper bounded by a \emph{cubic} polynomial of the vertex complexity and \emph{vice versa} (see, e.g., \cite{Schrijver03}). To sum up, any facet of polytope $\tilde{\mathcal{P}_2}$ has bit complexity $\O(n^3 + \ell)$, and therefore any vertex of $\tilde{\mathcal{P}_2}$ has $\O(n^9 \ell^3)$ bit complexity.

Let the polynomial $B(n, \ell) = O(n^9 \ell^3)$ be an upper bound on the maximum bit complexity of vertices of $\tilde{\mathcal{P}_2}$. Now we are ready to set the value of $\epsilon$. LP \eqref{lp:Usigma+-} always has an optimal vertex solution which we denote as $(u^*,\M^*)$. Recall that $u \leq -\frac{1}{2}$ for all points $(u,\M(p,\bvec{y}))$ in $\tilde{\mathcal{P}_2}$ and $u = -\frac{1}{2}$ is attainable at some vertices. Since all vertices of $\tilde{\mathcal{P}_2}$ have $B(n, \ell)$ bit complexity,  $(u^*,M^*)$ must either satisfy either $u^* = -\frac{1}{2}$ or $u^* \leq -\frac{1}{2} - 2^{-B(n,\ell)}$. Therefore, it suffices to set  $\epsilon = 2^{- n \cdot B(n,\ell)}$, which is a number with polynomial bit complexity. As a result, any optimal vertex solution to LP \eqref{lp:Usigma+-} must  satisfy $u^* = -\frac{1}{2}$, since the loss incurred by moving to any other vertex with $u < -\frac{1}{2}$ can never be compensated for by the other term $p\epsilon < \eps$.

With such a small value of $\epsilon$, the sender's goal is to send signal $\sigma_+$ with probability as high as possible, subject to the constraint that her utility from $\sigma_-$ is precisely $-\frac{1}{2}$. In other words, signal $\sigma_+$ must induce expected receiver/sender utility precisely $0$ for each regular action $i \neq 0$ (see Lemma \ref{lem:SenderUpperBound}). This characterization of the optimal scheme now allows us to determine whether $\M(\bvec{x}) \in \mathcal{K}$ by inspecting the sender's optimal expected utility. The following Lemma completes our proof of Theorem \,\ref{thm:hardness}.

\begin{lemma}\label{lem:MembershipJudge}
	Given the small enough value of $\epsilon$ described above, the sender's expected utility in the optimal signaling scheme for our constructed persuasion instance is at least $\frac{1}{2}(\epsilon-1)$ \emph{if and only if} $\M(\bvec{x}) \in \mathcal{K}$.
\end{lemma}
\begin{proof}
	$\Leftarrow$: If $\M(\bvec{x}) \in \mathcal{K}$, then by our choice of $a_i,b_i$ (recall Equations \eqref{eq:ConstructCond1} and \eqref{eq:ConstructCond2}), the signaling scheme implementing $\M(\bvec{x})$ is incentive compatible,  the sender's payoff from signal $\sigma_+$ is $\frac{1}{2} \epsilon$, and her payoff from $\sigma_-$ is  $-\frac{1}{2}$. Therefore, the optimal sender utility is at least  $\frac{1}{2} \epsilon - \frac{1}{2} $.
	
	$\Rightarrow$: Let $\M (p,\bvec{y})$ be the signature of a vertex optimal signaling scheme in LP \eqref{lp:Usigma+-}. By our choice of $\epsilon$ we know that the sender payoff from signal $\sigma_-$ must be exactly $ - \frac{1}{2}$. Therefore, to achieve overall sender utility at least $\frac{1}{2} \epsilon - \frac{1}{2}$, signal $\sigma_+$ must be sent with probability $p \geq \frac{1}{2}$, and the receiver's payoff from each regular action $i \neq 0$ in signal $\sigma_+$ is exactly $0$. That is, $y_i a_i + (p-y_i)b_i = 0$. By construction, we also have that $x_i a_i + (0.5 - x_i)b_i=0$  and $a_i,b_i \neq 0$, which imply that $\frac{y_i}{x_i} = \frac{p-y_i}{0.5 - x_i}$ and, furthermore, that $y_i \geq x_i$ since $p \geq \frac{1}{2}$. Now let $\varphi$ be a signaling scheme with the signature $\M(p,\bvec{y})$. We can post-process $\varphi$ so it has signature  $\M(\bvec{x})$ as follows: whenever $\varphi$ outputs the  signal $\sigma_+$, flip a biased random coin to output $\sigma_+$  with probability  $\frac{0.5}{p}$  and output $\sigma_-$ otherwise.  By using the identity $\frac{y_i}{x_i} = \frac{p-y_i}{0.5 - x_i}$, it is easy to see that this adjusted signaling scheme has signature $\M(\bvec{x})$.
\end{proof}

\newpage

\section{Omitted Proofs from Section \ref{sec:general}}
\label{app:general}

\subsection{A Bicriteria FPTAS}
\label{app:general:fptas}

\subsubsection*{Proof of Lemma \ref{lem:blackbox:ic}}
Fix $\eps$, $K$, and $\lambda$, and let $\varphi$ denote the resulting signaling scheme implemented by Algorithm~\ref{alg:blackbox}. Let $\theta \sim \lambda$ denote the input to $\varphi$, and $\sigma \sim \varphi(\theta)$ denote its output. First, we condition on the empirical sample $\tilde{\lambda} = \set{\theta_1,\ldots,\theta_K}$ without conditioning on the index $\ell$ of the input state of nature $\theta$, and show that $\eps$-incentive compatibility holds subject to this conditioning. The principle of deferred decisions implies that, subject to this conditioning, $\theta$ is uniformly distributed in  $\tilde{\lambda}$. By definition of linear program \eqref{lp:blackbox:empirical}, the signaling scheme $\tilde{\varphi}$ computed in Step \ref{step:lpempirical} is $\eps$-incentive compatible scheme for the empirical distribution $\tilde{\lambda}$. Since $\sigma \sim \tilde{\varphi}(\theta)$ and $\theta$ is conditionally distributed according to $\tilde{\lambda}$, this implies that all $\eps$-incentive compatibility constraints conditionally hold; formally, the following holds for each pair of actions $i$ and $j$:
\[ \Ex [ r_i(\theta) | \sigma=\sigma_i , \tilde{\lambda} ] \geq \Ex [ r_j(\theta) | \sigma=\sigma_i , \tilde{\lambda} ]  - \eps\]

Removing the conditioning on $\tilde{\lambda}$ and invoking linearity of expectations shows that $\varphi$ is $\eps$-incentive compatible for $\lambda$, completing the proof.

\subsubsection*{Proof of Lemma \ref{lem:blackbox:util}}
As in the proof of Lemma \ref{lem:blackbox:ic}, we condition on the empirical sample $\tilde{\lambda} = \set{\theta_1,\ldots,\theta_K}$ and observe that $\theta$ is uniformly distributed in $\tilde{\lambda}$ after this conditioning. The conditional expectation of sender utility then equals $\sum_{k=1}^K \sum_{i=1}^n \frac{1}{K} \tilde{\varphi}(\theta_k,\sigma_i) s_i(\theta_k)$, where $\tilde{\varphi}$ is the signaling scheme computed  in Step \ref{step:lpempirical} based on $\tilde{\lambda}$. Since this is precisely the optimal value of the LP \eqref{lp:blackbox:empirical} solved in Step \ref{step:lpempirical}, removing the conditioning and invoking linearity of expectations completes the proof.

\subsubsection*{Proof of Lemma \ref{lem:blackbox:apx}}
Recall that linear program \eqref{lp:persuasion} solves for the optimal incentive compatible scheme for $\lambda$. It is easy to see that the linear program \eqref{lp:blackbox:empirical} solved in step \ref{step:lpempirical} is simply the instantiation of LP \eqref{lp:persuasion} for the empirical distribution $\tilde{\lambda}$ consisting of $K$ samples from $\lambda$. To prove the lemma, it would suffice to show that the optimal incentive-compatible scheme $\varphi^*$ corresponding to LP \eqref{lp:persuasion} remains $\eps$-incentive compatible  and $\eps$-optimal for the distribution $\tilde{\lambda}$, with high probability. Unfortunately, this approach fails because polynomially-many samples from $\lambda$ are not sufficient to approximately preserve the incentive compatibility constraints corresponding to low-probability signals (i.e., signals which are output with probability smaller than inverse polynomial in $n$). Nevertheless, we show in Claim \ref{claim:probsignals} that there exists an approximately optimal solution $\hat{\varphi}$ to LP \eqref{lp:persuasion} with the property that every signal $\sigma_i$ is either \emph{large}, which we define as being output by $\hat{\varphi}$ with probability at least $\frac{\eps}{4n}$ assuming $\theta \sim \lambda$, or \emph{honest} in that only states of nature $\theta$ with $i \in \argmax_j r_j(\theta)$ are mapped to it.  It is easy to see that sampling preserves incentive-compatibility exactly for honest signals. As for large signals, we employ tail bounds and the union bound to show that  polynomially many samples suffice to approximately preserve incentive compatibility (Claim \ref{claim:preservesignals}).

\begin{claim}\label{claim:probsignals}
	There is a signaling scheme $\hat{\varphi}$ which is incentive compatible for $\lambda$,  induces sender utility $u_s(\hat{\varphi},\lambda) \geq OPT - \frac{\eps}{2}$ on $\lambda$, and such that every signal of $\hat{\varphi}$ is either large or honest.
\end{claim}
\begin{proof}
	Let $\varphi^*$ be the optimal incentive-compatible scheme for $\lambda$ --- i.e. the optimal solution to LP \eqref{lp:persuasion}.  We call a signal $\sigma$ \emph{small} if it is output by $\varphi^*$ with probability less than $\frac{\eps}{4n}$, i.e. if $\sum_{\theta \in \Theta} \lambda_\theta \varphi^*(\theta,\sigma) < \frac{\eps}{4n}$,  and otherwise we call it \emph{large}. Let $\hat{\varphi}$ be the scheme which is defined as follows: on input $\theta$, it first samples $\sigma \sim \varphi^*(\theta)$; if $\sigma$ is large then $\hat{\varphi}$ simply outputs $\sigma$, and otherwise it recommends an action maximizing receiver utility in state of nature $\theta$ ---- i.e., outputs $\sigma_{i'}$ for $i' \in \argmax_i r_i(\theta)$. It is easy to see that every signal of $\hat{\varphi}$ is either large or honest. Moreover, since $\varphi^*$ is incentive compatible and $\hat{\varphi}$ only replaces recommendations of $\varphi^*$ with ``honest'' recommendations, it is easy to check that $\hat{\varphi}$ is incentive compatible for $\lambda$. Finally, since the total probability of small signals in $\varphi^*$ is at most $\frac{\eps}{4}$, and utilities are in $[-1,1]$, the sender's expected utility from $\hat{\varphi}$ is no worse than $\frac{\eps}{2}$ smaller than her expected utility from $\varphi^*$.
\end{proof}

\begin{claim}\label{claim:preservesignals}
	Let $\hat{\varphi}$ be the signaling scheme from Claim \ref{claim:probsignals}. With probability at least $1-\frac{\eps}{8}$ over the sample $\tilde{\lambda}$, $\hat{\varphi}$  is $\eps$-incentive compatible for $\tilde{\lambda}$, and moreover $u_s(\hat{\varphi}, \tilde{\lambda}) \geq u_s(\hat{\varphi},\lambda) - \frac{\eps}{4}$. 
\end{claim}
\begin{proof}
	Recall that $\hat{\varphi}$ is incentive compatible for $\lambda$, and every signal is either large or honest. Since $\tilde{\lambda}$ is a set of samples from $\lambda$, it is easy to see that incentive compatibility constraints pertaining to the honest signals continue to hold over $\tilde{\lambda}$. It remains to show that incentive compatibility constraints for large signals, as well as expected sender utility, are approximately preserved when replacing $\lambda$ with $\tilde{\lambda}$. 
	
	Recall that incentive-compatibility requires that $\Ex_{\theta}[ \hat{\varphi}(\theta,\sigma_i) (r_i(\theta) - r_j(\theta))] \geq 0$ for each $i,j \in [n]$. Moreover, the sender's expected utility can be written as $\Ex_{\theta} [\sum_{i=1}^n \hat{\varphi}(\theta,\sigma_i) s_i(\theta)]$. The left hand side of each incentive compatibility constraint evaluates the expectation of a fixed function of $\theta$ with range $[-2,2]$, whereas the sender's expected utility evaluates the expectation of a function of $\theta$ with range in $[-1,1]$. Standard tail bounds and the union bound, coupled with our careful choice of the number of samples $K$, imply that replacing distribution $\lambda$ with $\tilde{\lambda}$ approximately preserves each of these $n^2+1$ quantities to within an additive error of $\frac{\eps^2}{4n}$ with probability at least $1-\frac{\eps}{8}$. This bound on the additive loss translates to $\eps$-incentive compatibility for the large signals, and is less than the   permitted decrease of $\frac{\eps}{4}$ for expected sender utility.
\end{proof}

The above claims, coupled with the fact that sender payoffs are bounded in $[-1,1]$, imply that the expected optimal value of linear program \eqref{lp:blackbox:empirical} is at least $OPT- \eps$, as needed.

\subsection{Information-Theoretic Barriers}
\label{app:general:barriers}

\subsection*{Impossibility of Incentive Compatibility (Proof of Theorem\,\ref{thm:hardness:blackbox} (a))}
Consider a setting with two states of nature, which we will conveniently refer to as \emph{rainy} and \emph{sunny}. The receiver, who we may think of as a daily commuter, has two actions: \emph{walk} and \emph{drive}. The receiver slightly prefers driving on a rainy day, and strongly prefers walking on a sunny day. We summarize the receiver's payoff function, parametrized by $\delta > 0$, in Table \ref{table:rainshine}. The sender, who we will think of as a municipality  with black-box sample access to weather reports drawn from the same distribution as the state of nature, strongly prefers that the receiver chooses walking regardless of whether it is sunny or rainy: we let $s_{walk}=1$ and $s_{drive}=0$ in both states of nature.

\begin{table}
	\small
	\begin{center}
		\begin{tabular}{ | c | c | c | }
			\hline
			& Rainy & Sunny \\ \hline 
			Walk & $1- \delta$ & 1 \\ \hline
			Drive & 1 & 0 \\
			\hline
		\end{tabular}
	\end{center}
	\caption{Receiver's Payoffs in Rain and Shine Example}
	\label{table:rainshine}
\end{table}

Let $\lambda_r$ be the point distribution on the rainy state of nature, and let $\lambda_s$ be such that $\Pr_{\lambda_s} [\mbox{rainy}] = \frac{1}{1+2\delta}$ and $\Pr_{\lambda_s} [\mbox{sunny}] =\frac{2 \delta}{1+2 \delta}$. It is easy to see that the unique direct incentive-compatible scheme for $\lambda_r$ always recommends driving, and hence results in expected sender utility of $0$. In contrast, a simple calculation shows that always recommending walking is incentive compatible for $\lambda_s$, and results in expected sender utility $1$. If algorithm $\A$ is incentive compatible and $c$-optimal for a constant $c<1$, then $\A(\lambda_r)$ must never recommend walking whereas $\A(\lambda_s)$ must recommend walking with constant probability at least $(1-c)$ overall (in expectation over the input state of nature $\theta \sim \lambda_s$  as well as all other internal randomness). Consequently, given a black box distribution $\D \in \set{\lambda_r,\lambda_s}$, evaluating $\A(\D,\theta)$ on a random draw $\theta \sim \D$ yields a tester which distinguishes between $\lambda_r$ and $\lambda_s$ with constant probability~$1-c$.

Since the total variation distance between $\lambda_r$ and $\lambda_s$ is $O(\delta)$, it is well known (and easy to check) that any black-box algorithm which distinguishes between the two distributions with $\Omega(1)$ success probability must take $\Omega(\frac{1}{\delta})$ samples in expectation when presented with one of these distributions. As a consequence, the average-case sample complexity of $\A$ on either of $\lambda_r$ and $\lambda_s$ is $\Omega(\frac{1}{\delta})$. Since $\delta > 0$ can be made arbitrarily small, this completes the proof.

\subsection*{Impossibility of Optimality (Proof of Theorem\,\ref{thm:hardness:blackbox} (b))}

Consider a setting with three actions $\set{1,2,3}$ and three corresponding states of nature $\theta_1, \theta_2, \theta_3$. In each state $\theta_i$, the receiver derives utility $1$ from action $i$ and utility $0$ from the other actions. The sender, on the other hand, derives utility $1$ from action $3$ and utility $0$ from actions $1$ and $2$. For an arbitrary parameter $\delta > 0$, we define two distributions $\lambda$ and $\lambda'$ over states of nature with total variation distance $\delta$, illustrated in Table \ref{table:twodists}.

Assume algorithm $\A$ is optimal and $c$-incentive compatible for a constant $c < \frac{1}{4}$.   The optimal incentive-compatible scheme for $\lambda'$ results in expected sender utility $3\delta$ by recommending action $3$ whenever the state of nature is $\theta_2$ or $\theta_3$, and with probability $\frac{\delta}{1-2 \delta}$ when the state of nature is $\theta_1$. Some calculation reveals that in order to match this expected sender utility subject to $c$-incentive compatibility, signaling scheme $\varphi'=\A(\lambda')$ must satisfy $\varphi'(\theta_2, \sigma_3) \geq \mu$ for $\mu =  1-4c > 0$. In other words, $\varphi'$ must recommend action $3$ a constant fraction of the time when given state $\theta_2$ as input. In contrast, since $c < \frac{1}{2}$ it is easy to see that $\varphi=\A(\lambda)$ can never recommend action $3$: for any signal, the posterior expected receiver reward for action $3$ is $0$, whereas one of the other two actions must have posterior expected receiver reward at least $\frac{1}{2}$. It follows that given $D \in \set{\lambda,\lambda'}$, a call to $\A(\D,\theta_2)$ yields a tester which distinguishes between $\lambda$ and $\lambda'$ with constant probability $\mu$. Since $\lambda$ and $\lambda'$ have statistical distance $\delta$, we conclude that the worst case sample complexity of $\A$ on either of $\lambda$ or $\lambda'$ is $\Omega(\frac{1}{\delta})$.  Since $\delta > 0$ can be made arbitrarily small, this completes the proof.

\begin{table}
	\small
	\begin{center}
		\begin{tabular}{ | c | c | c | c | }
			\hline
			& $\Pr[\theta_1]$ & $\Pr[\theta_2]$ & $\Pr[\theta_3]$ \\ \hline 
			$\lambda$ & $1- 2 \delta$ & $2 \delta$ & 0 \\ \hline
			$\lambda'$ & $1- 2 \delta$ & $\delta$ & $\delta$ \\
			\hline
		\end{tabular}
	\end{center}
	\caption{Two Distributions on Three Actions}
	\label{table:twodists}
\end{table}

\end{document}